\documentclass[12pt,reqno]{amsart}
\usepackage{float}
\usepackage{textcomp}
\usepackage{geometry} 
\usepackage{comment}

\usepackage{indentfirst}
\usepackage{amsmath}
\usepackage{amsfonts}
\usepackage{amssymb}
\usepackage{amsthm}

\usepackage{subcaption}
\usepackage{graphicx} 
\usepackage{wrapfig} 

\usepackage{mathtools}

\usepackage{booktabs}
\usepackage{enumerate}
\usepackage{enumitem}
\setlist[itemize]{noitemsep} 

\usepackage{xcolor}

\usepackage{hyperref} 
\usepackage{indentfirst}
\numberwithin{equation}{section}

\usepackage{amsthm}

\usepackage{amsthm}
\newtheorem{prop}{Proposition}[section]

\newtheorem{lem}[prop]{Lemma}

\newtheorem{conj}[prop]{Conjecture}
\newtheorem{cor}[prop]{Corollary}
\newtheorem{thm}[prop]{Theorem}

\theoremstyle{remark}
\newtheorem{rem}{Remark}

\newcommand{\bC}{\mathbb{C}}

\newcommand{\bZ}{\mathbb{Z}}
\newcommand{\bN}{\mathbb{N}}

\newcommand{\ii}{\infty}
\newcommand{\ap}{\alpha}
\newcommand{\be}{\beta}
\newcommand{\de}{\delta}

\newcommand{\ga}{\gamma}

\allowdisplaybreaks[4]

\begin{document}

\title[Proofs of some conjectures on line defect half-indices]{Proofs of some conjectures of Okazaki and Smith on line defect half-indices of ${\rm SU}(N)$ Chern--Simons theories}
\author{Liuquan Wang and Yiyang Yue}
\address[L.\ Wang]{School of Mathematics and Statistics, Wuhan University, Wuhan 430072, Hubei, People's Republic of China}
\email{wanglq@whu.edu.cn;mathlqwang@163.com}

\address[Y.\ Yue]{School of Mathematics and Statistics, Wuhan University, Wuhan 430072, Hubei, People's Republic of China}
\email{yiyangyue@whu.edu.cn}



\begin{abstract}
Okazaki and Smith discovered many elegant formulas expressing some matrix integrals  as some celebrated $q$-series such as the Rogers--Ramanujan functions or Jacobi theta functions. These integrals arise as Wilson line defect half-indices of 3d $\mathcal{N}=2$ supersymmetric ${\rm SU}(N)$ Chern--Simons theories. We evaluate them by carefully calculating the constant terms of some infinite products. Along the way we use  some crucial facts about antisymmetric multivariate formal Laurent series. Consequently, we prove three general conjectures of Okazaki and Smith which provide explicit formulas for half indices of the ${\rm SU}(N)_{-N-k}$ ($k=0,1/2,1$) Chern-Simons theories. During the process, we extend these ${\rm SU}(N)$ formulas to include one additional parameter. Furthermore, we generalize the ${\rm SU}(N)_{-N-1/2}$ and ${\rm SU}(N)_{-N-1}$ conjectures by calculating the corresponding half-indices of Wilson lines of arbitrary charge. As a special instance of our generalizations, we also confirm the ${\rm SU}(3)_{-4}$ conjecture of Okazaki and Smith.
\end{abstract}

\maketitle
\tableofcontents

\section{Introduction and main results}\label{sec-intro}
In 3d $\mathcal{N}=2$ supersymmetric field theories with boundary, half-indices are one of the main tools for analyzing and verifying dualities of bulk-boundary systems, which only depend on boundary conditions of the given field and can be written as a specific supersymmetric partition function. For example, it can be used to examine the spectra of BPS local operators in the bulk-boundary system if the given field is with the BPS boundary conditions.

 As a generalization of half-indices, field theorists studied line defect half-indices, that is, the corresponding half-index is considered on the boundary with the endpoints of line defects. Then the corresponding line defect operator should be inserted into the half-index formula in this case. Similarly to  half-indices, line defect half-indices can characterize the degrees of freedom of the boundary and defect, and they are important tools to test the duality of ``bulk-boundary-line'' systems. By studying line defect half-indices, we can explore the spectra of line defects in 3d bulk theories and the structure of the induced boundary.

The theory of $q$-series played an important role in the evaluation of half-indices. For instance, one may find integral representations for half-indices which provide interpretations of some important identities related to root systems including the Askey--Wilson type $q$-beta integrals \cite{Askey-Wilson}. Recent works of Okazaki and Smith \cite{OS-23JHEP,OS24,OS-24JHEP} provide many such examples. In particular, they found that the line defect half-indices for the 3d $\mathcal{N}=2$ supersymmetric ${\rm SU}(N)$ Chern--Simons (CS) theories with Neumann boundary conditions for gauge fields can be elegantly represented by some well-known $q$-series functions. The half-indices are represented as multi-dimensional matrix integrals, and they discovered many interesting identities expressing  such integrals as elegant $q$-series or eta products.

We will focus on evaluations of half-indices of ${\rm SU}(N)$ CS theories in  the work of Okazaki and Smith \cite{OS-24JHEP}, and we refer the reader to \cite{OS-24JHEP} for details of their physical interpretations. Let $N_f$ and $N_a$ denote the number of fundamental and anti-fundamental 3d chiral multiplets with Neuman boundary conditions with $R$-charge 1, respectively, and let $M$ denote the number of fundamental 2d Fermi multiplets with $R$-charge 0.  To cancel the gauge anomaly, we take the CS level $-N-k$ where \cite[(2.2)]{OS-24JHEP}
\begin{equation}
    k=-\frac12(N_f+N_a)+M. \label{CS-level SUN-N 2}
\end{equation}
The initial three basic cases $k=0,1/2$ and $1$ are particularly interesting,  as we will soon see that there are elegant formulas for the half-indices defined by
\begin{align}
        &\mathbb{II}_{\mathcal{N}}^{{\rm SU}(N)_{-N}}(q) :=\frac{(q;q)_{\ii}^{N-1}}{N!} \oint \bigg( \prod_{i=1}^{N-1} \frac{ds_i}{2\pi is_i} \bigg) \prod_{i \ne j}^N (s_is_j^{-1};q)_{\ii}, \\
        & \mathbb{II}_{\mathcal{N}}^{{\rm SU}(N)_{-N-1/2}}(x;q) :=  \frac{(q;q)_{\ii}^{N-1}}{N!} \oint \prod_{i=1}^{N-1} \frac{ds_i}{2\pi is_i} \prod_{i \ne j}^N (s_is_j^{-1};q)_{\ii} \prod_{i=1}^N (q^{\frac12}s_ix;q)_{\ii},\\
       & \mathbb{II}_{\mathcal{N}}^{{\rm SU}(N)_{-N-1}}(x;q):=  \frac{(q;q)_{\ii}^{N-1}}{N!} \oint \bigg( \prod_{i=1}^{N-1} \frac{ds_i}{2\pi is_i} \bigg) \prod_{i \ne j}^N (s_is_j^{-1};q)_{\ii} \prod_{i=1}^N(q^{\frac12}s_i^{\pm} x^{\pm};q)_{\ii}.
\end{align}
Here the notation $(a^{\pm};q)_\infty$ stands for $(a,a^{-1};q)_\infty$ where we use the following standard $q$-series notation:
\begin{align}
(a_1,a_2,\dots,a_m;q)_n:=\prod\limits_{i=1}^m\prod\limits_{k=0}^{n-1}(1-a_iq^k), \quad |q|<1, \quad n\in \mathbb{N}\cup \{\infty\}.
\end{align}

By introducing BPS line operators that are perpendicular to the boundary, one can modify the half-index to give line defect correlators. In particular, if we introduce the Wilson line operator $W_\mathcal{R}$ associated with the representation $\mathcal{R}$ of the gauge group, such line defect correlator is calculated by inserting the associated character $\chi_\mathcal{R}$ into the matrix integral. In the case of ${\rm SU}(N)$ we take a basis for the set of Wilson lines $W_m$ labeled by the power symmetric functions of degree $m$
\begin{align}
p_m(s) = \sum_{i=1}^Ns_i^m
\end{align}
under the restriction  $\prod\limits_{i=1}^N s_i=1$. Following \cite{OS24,OS-24JHEP}, we call the Wilson line labeled by $p_m(s)$ the charge-$m$ Wilson line, and we denote the half-index in the presence of such a Wilson line by $\langle W_m\rangle$. In particular, the Wilson line corresponding to ${\rm SU}(N)_{-N-k}$ with $k\in \{0,1/2,1\}$ is given by
\begin{align}\label{intro-Wm}
   &\langle W_m\rangle^{{\rm SU}(N)_{-N-k}}(x;q):=\frac{(q;q)_{\ii}^{N-1}}{N!} \oint \bigg( \prod_{i=1}^{N-1} \frac{ds_i}{2\pi is_i} \bigg) \prod_{i\neq j}^N (s_is_j^{-1};q)_{\ii} \times \Big( \prod_{i=1}^N(q^{\frac12}s_i x;q)_{\ii} \Big)^{\epsilon_k}  \nonumber \\
    &\qquad \qquad \qquad \qquad \qquad \qquad  \times \Big(\prod_{i=1}^N(q^{\frac12}s_i^{-1} x^{-1};q)_{\ii}\Big)^{\delta_k} \times \sum_{i=1}^Ns_i^m
\end{align}
where $(\epsilon_k,\delta_k)=(0,0)$ for $k=0$, $(1,0)$ for $k=1/2$ and $(1,1)$ for $k=1$. It follows by definition that for $k\in \{0,1/2,1\}$,
\begin{align}
    \langle W_0\rangle^{{\rm SU}(N)_{-N-k}}=N  \mathbb{II}_{\mathcal{N}}^{{\rm SU}(N)_{-N-k}}.
\end{align}

An easy but useful observation is that evaluating half-indices amounts to calculating constant terms of the integrands. Using Euler's $q$-exponential identities (see e.g.~\cite[Corollary 2.2]{Andrews-book})
\begin{align}\label{euler1}
 \frac{1}{(z;q)_{\ii}} = \sum_{n=0}^{\ii} \frac{z^n}{(q)_n} \quad (|z|<1), \quad (-z;q)_{\ii} = \sum_{n=0}^{\infty}\frac{z^nq^{n(n-1)/2}}{(q)_n}
\end{align}
and the Jacobi triple product identity (see e.g.~\cite[Theorem 2.8]{Andrews-book})
\begin{align}\label{jtpi1}
     (q;q)_{\infty}(z^{-1};q)_{\infty}(qz;q)_{\infty} = \sum_{m=-\infty}^{\infty}(-1)^mq^{m(m+1)/2}z^m,
\end{align}
we can expand the integrand in \eqref{intro-Wm} as formal Laurent series in the variables $s_1^{\pm},\dots,s_{N-1}^{\pm}$.
Throughout this paper, for any formal power series $F(x)$  in $\bC[[x_1^{\pm},\dots,x_n^{\pm}]]$ and any $\ap = (\ap_1,\dots,\ap_n) \in \bZ^n$, we use $[x^\ap ]F$ to denote the coefficient of $x^{\ap} = x_1^{\ap_1}\cdots x_n^{\ap_n}$ in $F$. Then the half-index associated with the Wilson line can be expressed as a constant term of the form:
\begin{align}\label{eq-CT}
    \langle W_m\rangle^{{\rm SU}(N)_{-N-k}}(x;q)&=\frac{(q;q)_{\ii}^{N-1}}{N!}  [(s_1s_2\cdots s_{N-1})^0] \prod_{i\neq j}^N (s_is_j^{-1};q)_{\ii} \times \Big( \prod_{i=1}^N(q^{\frac12}s_i x;q)_{\ii} \Big)^{\epsilon_k}  \nonumber \\
    &\qquad \times \Big(\prod_{i=1}^N(q^{\frac12}s_i^{-1} x^{-1};q)_{\ii}\Big)^{\delta_k} \times \sum_{i=1}^Ns_i^m.
\end{align}

Okazaki and Smith \cite{OS-24JHEP} first considered the line defect half-indices for the 3d $\mathcal{N}=2$ supersymmetric $SU(2)$ CS theory with Wilson line operators and Neumann boundary conditions for the vector multiplet. Among other results, they obtained the following: \\
(1) For $k=-2$, we have \cite[Section 3.1]{OS-24JHEP}
    \begin{align}\label{eq-SU22}
  \langle W_{2m}\rangle^{{\rm SU}(2)_{-2}}(q)=(-1)^m\big(q^{m(m-1)/2}+q^{m(m+1)/2}\big), \quad  \langle W_{2m+1}\rangle^{{\rm SU}(2)_{-2}}(q)=0.
    \end{align}
(2) For $k=-5/2$ we have \cite[(3.16) and (3.27)]{OS-24JHEP}
    \begin{align}
        \mathbb{II}_{\mathcal{N}}^{{\rm SU}(2)_{-5/2}}(x;q)&=\sum_{n=0}^\infty \frac{q^{n^2}x^{2n}}{(q;q)_n}, \label{intro-SU252-Pi} \\
        \langle W_1\rangle^{{\rm SU}(2)_{-5/2}}(x;q)&=-q^{\frac{1}{2}}\sum_{n=0}^\infty \frac{q^{n(n+1)}x^{2n+1}}{(q;q)_n}. \label{intro-SU252-W1}
    \end{align}
The two functions on the right side are essentially the famous Rogers--Ramanujan functions. They only provided proof for the $x=1$ case. \\
(3) For $k=-3$ they gave explicit calculations for $\langle W_k \rangle^{{\rm SU}(2)_{-3}}(x;q)$ and found that the results depend on $k$ modulo 6. In particular, we have \cite[(3.66) and (3.71)]{OS-24JHEP}
    \begin{align}
        \mathbb{II}_{\mathcal{N}}^{{\rm SU}(2)_{-3}}(x;q)&=\frac{1}{(q;q)_\infty}\sum_{n\in \mathbb{Z}}q^{n^2}x^{2n}, \label{SU23-Pi} \\
        \langle W_1\rangle^{{\rm SU}(2)_{-3}}(x;q)&=-\frac{1}{(q;q)_\infty}\sum_{n\in \mathbb{Z}}q^{n^2+n+\frac{1}{2}}x^{2n+1}. \label{SU23-W1}
    \end{align}

They also studied line defect half-indices for the 3d $\mathcal{N}= 2$ ${\rm SU}(3)$ CS theory with Wilson lines and Neumann boundary conditions for the vector multiplet.\\
(1) For $k=-3$  it was proved that \cite[(4.11)]{OS-24JHEP}
    \begin{align}\label{SU33-Wk}
        \langle W_m\rangle^{{\rm SU}(3)_{-3}}(q)=\left\{\begin{array}{ll}
q^{\frac{m(m-3)}{9}}\Big(1+q^{\frac{m}{3}}+q^{\frac{2m}{3}}\Big), & m\equiv 0 \pmod{3}, \\
0, & \text{otherwise}.
        \end{array}\right.
    \end{align}
(2) For $k=-7/2$ it was found without proof that \cite[(4.17) and (4.19)]{OS-24JHEP}
    \begin{align}
        \mathbb{II}_{\mathcal{N}}^{{\rm SU}(3)_{-7/2}}(x;x^{-1};q)&=\sum_{n=0}^\infty \frac{(-1)^nq^{\frac{3}{2}n^2}x^{3n}}{(q;q)_n}, \\
        \langle W_1\rangle^{{\rm SU}(3)_{-7/2}}(x;x^{-1};q)&=\sum_{n=0}^\infty \frac{(-1)^nq^{\frac{3}{2}n^2+2n+1}x^{3n+2}}{(q;q)_n}.
    \end{align}
(3) For $k=-4$, it was found without proof that \cite[(4.23), (4.27) and (4.29)]{OS-24JHEP}
    \begin{align}
        \mathbb{II}_{\mathcal{N}}^{{\rm SU}(3)_{-4}}(x;q)&=\frac{1}{(q;q)_\infty}\sum_{n\in \mathbb{Z}}(-1)^nq^{\frac{3}{2}n^2}x^{3n}, \\
        \langle W_1\rangle^{{\rm SU}(3)_{-4}}(x;q)&=\frac{1}{(q;q)_\infty}\sum_{n\in \mathbb{Z}}(-1)^nq^{\frac{3}{2}n^2+2n+1}x^{3n+2}, \\
         \langle W_2\rangle^{{\rm SU}(3)_{-4}}(x;q)&=-\frac{1}{(q;q)_\infty}\sum_{n\in \mathbb{Z}}(-1)^nq^{\frac{3}{2}n^2+2n+1}x^{-3n-2}.  \label{SU34-W2}
    \end{align}
    They also found that $\langle W_m\rangle^{{\rm SU}(3)_{-4}}$ ($3\leq m\leq 12$) can be expressed by $\mathbb{II}_{\mathcal{N}}^{{\rm SU}(3)_{-4}}$, $\langle W_1\rangle^{{\rm SU}(3)_{-4}}$ and $\langle W_2\rangle^{{\rm SU}(3)_{-4}}$. Based on these formulas, they \cite[Section 4.3.1, p.~29]{OS-24JHEP} made the following conjecture.
\begin{conj}\label{conj-SU34}
The general case of $\langle W_m\rangle_{\mathcal{N},1,0}^{{\rm SU}(3)_{-4}}$ is given by the following:
\begin{itemize}
    \item A factor $\left\{\begin{array}{ll}
    \mathbb{II}_{\mathcal{N}}^{{\rm SU}(3)_{-4}}, &m\equiv 0 \pmod{3}, \\
    \langle W_1\rangle^{{\rm SU}(3)_{-4}}, & m\equiv 1\pmod{3}, \\
    \langle W_2\rangle^{{\rm SU}(3)_{-4}}, &m\equiv 2 \pmod{3},\end{array}\right.$
    \item A factor $(1+q^{m/4}+q^{m/2})$ for $m\equiv 0 \pmod{4}$ and $m\neq 0$,
    \item A factor of a single power of $q$ and a sign $\pm 1$.
\end{itemize}
\end{conj}

Okazaki and Smith \cite[Section 5]{OS-24JHEP} then proposed conjectures for half-indices for the general ${\rm SU}(N)$ gauge group and numerically verified them for some $N$. We will prove all the conjectures with CS level $-N-k$ ($k=0,1/2,1$) and restate them as the following theorems. For the case ${\rm SU}(N)_{-N}$ the following formula was conjectured in  \cite[(5.1)]{OS-24JHEP}.
\begin{thm}[The ${\rm SU}(N)_{-N}$ conjecture]\label{thm-SUN-N}
For any integer $m$ we have
    \begin{equation}
    \begin{aligned}
        &\langle W_m\rangle^{{\rm SU}(N)_{-N}}(q)
        = \begin{cases}
        (-1)^{\ell(N-1)} q^{(N-1)\ell(\ell-1)/2} \frac{1-q^{\ell N}}{1-q^\ell}, & m=\ell N, \\
        0, & m \ne \ell N,
    \end{cases} \quad \ell \in \bZ.
    \end{aligned}
    \end{equation}
\end{thm}
For the case ${\rm SU}(N)_{-N-1/2}$, Okazaki and Smith \cite[(5.18) and (5.21)]{OS-24JHEP} conjectured  the following formulas.
\begin{thm}[The ${\rm SU}(N)_{-N-1/2}$ conjecture]\label{thm-SUN-N-12}
We have
\begin{align}
  & \mathbb{II}_{\mathcal{N}}^{{\rm SU}(N)_{-N-1/2}}(x;q)= \sum_{n=0}^{\ii} \frac{(-1)^{Nn} q^{\frac{Nn^2}{2}}x^{Nn}}{(q;q)_n}, \label{eq-thm-SUN-N-12-1}\\
      &\langle W_1\rangle^{{\rm SU}(N)_{-N-1/2}}(x;q) = \sum_{n=0}^{\ii} \frac{(-1)^{Nn+(N-1)} q^{\frac{Nn^2}{2}+(N-1)n + \frac{N-1}{2}}x^{Nn+(N-1)}}{(q;q)_n}. \label{eq-thm-SUN-N-12-2}
\end{align}
\end{thm}

Furthermore, we find that the general case $\langle W_m\rangle^{{\rm SU}(N)_{-N-1/2}}(x;q)$ can be evaluated as follows.
\begin{thm}\label{thm-SUN-N-12-general}
For any integer $m$ we have
\begin{align}
      \langle W_m\rangle^{{\rm SU}(N)_{-N-1/2}}(x;q)  =\sum_{k = \lceil m/N \rceil}^{\ii} \sum_{i=1}^N Y_{i,m,k}(x,q).  \label{eq-thm-SUN-N-12-Wn}
\end{align}
Here, $Y_{i,m,k}$ ($1\leq i \leq N$) is given as follows:
    \begin{align}
        Y_{i,0,k}(x,q) &= \begin{cases}
            \frac{(-x)^{Nk}q^{Nk^2/2}}{(q;q)_k }, & k \ge 0,\\
            0, & k < 0,
        \end{cases} \label{Y_i0k}
    \end{align}
for $1\leq m \leq N$ we have
\begin{align}
        Y_{i,m,k}(x,q) &= \begin{cases}
            \frac{(-1)^{Nk+1} x^{Nk-m} q^{(Nk^2+m)/2-(N+1-i)k}}{(q;q)_{k-1}}, & N-i+1 \le m \le N \text{ and } k > 0, \\
            0, & 0<m<N-i+1 \text{ or } k \le 0,
        \end{cases} \label{Y-imk}
    \end{align}
and for other $m$, $Y_{i,m,k}$ can be computed recursively:
    \begin{align}\label{rela-Y-12}
        Y_{i,m,k}(x,q) =\begin{cases}
        q^{\frac12}x^{-1} \big( Y_{i,m-1,k}(x,q) + (-1)^N q^{m+i-k-N-1} Y_{i,m-N-1,k-1}(x,q) \big), & m>N,  \\
        (-1)^{N-1} q^{k-i-m+1} \big( Y_{i,m+N,k+1}(x,q) - xq^{-\frac12} Y_{i,m+N+1,k+1}(x,q) \big), & m<0.
        \end{cases}
    \end{align}
\end{thm}

While $\langle W_{m}\rangle^{{\rm SU}(N)_{-N-1/2}}(x;q)$ can be calculated explicitly using Theorem \ref{thm-SUN-N-12-general} for each $m$, the formula becomes more and more complicated as $|m|$ increases. Here we present formulas for the case $1\leq |m|\leq N+1$ as examples.
\begin{cor}\label{cor-SUN-N-12-general}
For $1 \le m \le N$, we have
    \begin{align}
        \langle W_m\rangle^{{\rm SU}(N)_{-N-1/2}}(x;q) =(-1)^{N+1}(xq^{\frac12})^{N-m} \sum_{k=0}^\infty \frac{(-x)^{Nk} q^{Nk^2/2+(N-m)k}\big(1-q^{m(k+1)} \big)}{(q;q)_k(1-q^{k+1})}. \label{intro-Wn-positive}
    \end{align}
For $1-N\leq m \leq -1$, we have
\begin{align}
        \langle W_m\rangle^{{\rm SU}(N)_{-N-1/2}}(x;q)
    = -(xq^{\frac12})^{-m} \sum_{k=0}^\infty \frac{(-x)^{Nk} q^{Nk^2/2-mk}}{(q;q)_k}. \label{intro-Wn-negative}
    \end{align}
For $m=-N$, we have
    \begin{align}
         &\langle W_{-N}\rangle^{{\rm SU}(N)_{-N-1/2}}(x;q) = (-1)^{N+1} \sum_{k=0}^\infty \frac{(-x)^{Nk}q^{Nk^2/2}}{(q;q)_k} \bigg( q^k \frac{1-q^N}{1-q} + (-xq^{\frac12})^N q^{Nk} \bigg). \label{intro-Wn-negative-n}
    \end{align}
For $|m|=N+1$, we have
\begin{align}
 \langle W_{N+1}\rangle^{{\rm SU}(N)_{-N-1/2}}(x;q)  &=  (-1)^N\sum_{k=0}^\infty \frac{(-y)^{Nk-1} q^{(Nk^2+1)/2}}{(q;q)_k} \bigg( \frac{1-q^{N(k+1)}}{1-q^{k+1}} - q^{-k} \frac{1-q^N}{1-q} \bigg), \label{Cm>0-2} \\
        \langle W_{-N-1}\rangle^{{\rm SU}(N)_{-N-1/2}}(x;q)  &= (-1)^{N-1} \sum_{k=0}^\infty \frac{(-x)^{Nk+1}q^{(Nk^2+1)/2+k}}{(q;q)_k} \bigg( q^{k+N} + \frac{1-q^N}{1-q} \bigg). \label{Cm<0-2}
\end{align}
\end{cor}

For the case ${\rm SU}(N)_{-N-1}$ Okazaki and Smith \cite[(5.25) and (5.32)]{OS-24JHEP} conjectured the following formulas.
\begin{thm}[The ${\rm SU}(N)_{-N-1}$ conjecture]\label{thm-SUN-N-1}
We have
\begin{align}
 & \mathbb{II}_{\mathcal{N}}^{{\rm SU}(N)_{-N-1}}(x;q)= \frac{1}{(q;q)_{\ii}} \sum_{n\in\bZ} (-1)^{Nn} q^{\frac{Nn^2}{2}}x^{Nn}, \label{Pi-SUN-N-1} \\
   &\langle W_1\rangle^{{\rm SU}(N)_{-N-1}}(x;q)
        = \frac{(-xq^{\frac12})^{N-1}}{(q;q)_{\ii}} \sum_{n\in\bZ} (-1)^{Nn} q^{\frac{Nn^2}{2}+(N-1)n}x^{Nn}. \label{W1-SUN-N-1}
\end{align}
\end{thm}

Furthermore, we find elegant formulas for $\langle W_m\rangle^{{\rm SU}(N)_{-N-1}}(x;q)$ for an arbitrary charge $m$ and therefore give a generalization to Theorem \ref{thm-SUN-N-1}.
\begin{thm}\label{thm-SUN-N1-general}
    Let $m$ be an integer and $t=\lfloor m/(N+1) \rfloor$. We have
    \begin{align}\label{eq-thm-SUN-N1-general}
      & \langle W_m\rangle^{{\rm SU}(N)_{-N-1}}(x;q)   \\
     & = \begin{cases}
            \frac{1}{(q;q)_{\ii}} x^{-m} q^{(m-N)t+m/2} \frac{1-q^{Nt}}{1-q^t} \sum_{k\in\bZ} (-1)^{Nk} x^{Nk} q^{Nk^2/2 -mk}, & m \equiv 0 \pmod{N+1}, \\
            -\frac{1}{(q;q)_{\ii}} x^{-m} q^{mt+m/2} \sum_{k\in\bZ} (-1)^{Nk} x^{Nk} q^{Nk^2/2 -mk}, & m \not\equiv 0 \pmod{N+1}. \nonumber
        \end{cases}
    \end{align}
\end{thm}
In particular, when $N=3$, we get an explicit formula for $\langle W_m\rangle^{{\rm SU}(3)_{-4}}$ and confirm Conjecture \ref{conj-SU34}.
\begin{cor}[The ${\rm SU}(3)_{-4}$ conjecture]\label{cor-SU34-conj}
We have
\begin{align}\label{SU34-W}
&\langle W_m\rangle^{{\rm SU}(3)_{-4}}(x;q)  \\
&= \begin{cases}
    x^{2m} q^{\frac34m^2} (q^{-\frac{m}{4}} + 1 + q^{\frac{m}{4}}) (q^{\frac32-2m}x^{-3},q^{\frac32+2m}x^3,q^3;q^3)_\infty/(q;q)_\infty, & m \equiv 0 \pmod 4,\\
    x^{2m} q^{\frac34m^2 \pm \frac{m}{4}} (q^{\frac32-2m}x^{-3},q^{\frac32+2m}x^3,q^3;q^3)_\infty/(q;q)_\infty, & m \equiv \pm1 \pmod 4,\\
    -x^{2m} q^{\frac34m^2} (q^{\frac32-2m}x^{-3},q^{\frac32+2m}x^3,q^3;q^3)_\infty/(q;q)_\infty,  & m \equiv 2 \pmod 4.
\end{cases} \nonumber
\end{align}
As a consequence, Conjecture \ref{conj-SU34} holds.
\end{cor}

It should be emphasized that the proofs of the above theorems are much more difficult than the special cases discussed in \eqref{eq-SU22}--\eqref{SU33-Wk}. We briefly discuss the method in \cite{OS-24JHEP} to prove \eqref{eq-SU22}--\eqref{SU33-Wk}.   Okazaki and Smith \cite{OS-24JHEP} first expand the infinite products using the Jacobi triple product identity \eqref{jtpi1}, and then we can express the constant term in \eqref{eq-CT} as  some multi-sums. For some small $N$, these multi-sums can be further simplified by delicate analysis.  However, this method requires a case-by-case analysis and cannot be applied to the general case. If we expand the infinite products in the integrals of Theorems \ref{thm-SUN-N}--\ref{thm-SUN-N-12-general}, \ref{thm-SUN-N-1} and \ref{thm-SUN-N1-general} using the Jacobi triple product identity (for Theorem \ref{thm-SUN-N-12} we may also need Euler's $q$-exponential identity), we will get too many sums whose simplifications seem far from reach. Okazaki and Smith further commented that \cite[Section 1.2]{OS-24JHEP}
\begin{quote}
    It would be nice to give analytic proofs of our formulas for arbitrary ${\rm SU}(N)$ gauge group. For higher rank gauge groups, they will be obtained from a certain multivariable
extension of the Jacobi triple product identity.
\end{quote}
It is not clear to us whether such multivariable version of \eqref{jtpi1} exists or not. Nevertheless, we find the clues to prove these conjectures using a different approach inspired by Stembridge's work \cite{Stembridge1988}. In that work  Stembridge presented a short proof of Macdonald's conjecture \cite[(3.1)]{Macdonald-1982} for the root system $A_{n-1}$:
\begin{align}\label{eq-speical-q-Dyson}
 [x^0]\prod\limits_{1\leq i<j\leq n} (qx_ix_j^{-1};q)_{a}(x_jx_i^{-1};q)_a=\frac{(q;q)_{an}}{(q;q)_a^n}.
\end{align}
This is also the equal parameter case of Andrews' $q$-Dyson conjecture \cite{Andrews1975} proved by Zeilberger and Bressoud \cite{Zeilberger-Bressoud}.
Motivated by Stembridge's proof of \eqref{eq-speical-q-Dyson}, we carefully analyze certain ways of collecting coefficients in antisymmetric  formal Laurent series with $n$ variables and establish formulas to simplify them. Then for some specific series related to the infinite products in the integrands of Theorems \ref{thm-SUN-N}--\ref{thm-SUN-N-12-general} and \ref{thm-SUN-N-1}--\ref{thm-SUN-N1-general}, we are able to evaluate their constant terms and thus prove these ${\rm SU}(N)$ conjectures and evaluate the half-indices for Wilson lines with arbitrary charge. Furthermore, we will also generalize these ${\rm SU}(N)$ conjectures to involve an additional parameter (see Theorems \ref{thm-CT-1}, \ref{thm-CT-SU-N12} and \ref{thm-CT-SU-N1}).

The remainder of this paper is organized as follows. In Section \ref{sec-pre} we collect and establish some useful results about antisymmetric formal Laurent series in several variables.  We present proofs and generalizations of Theorems \ref{thm-SUN-N}--\ref{thm-SUN-N-12-general} and \ref{thm-SUN-N-1}--\ref{thm-SUN-N1-general} and Corollaries \ref{cor-SUN-N-12-general} and \ref{cor-SU34-conj} in Section \ref{sec-general}. Finally, we discuss some constant term identities as byproducts of the main results and their proofs and give several remarks in Section \ref{sec-rem}.

\section{Preliminaries}\label{sec-pre}

In this section, we collect and establish some important facts about antisymmetric polynomials with several variables.

\subsection{Some notations and basic facts}
We follow \cite{Stembridge1988} to introduce some basic facts about extracting certain terms of formal Laurent series. Let $S_n$ be the symmetric group of $[n] := \{1,2,\dots,n\}$. We say that a formal power series $F_n=\sum_{\alpha\in \mathbb{Z}^n} c_\alpha x^\alpha \in \mathbb{C}[[x_1^{\pm 1},\dots,x_n^{\pm 1}]]$ is symmetric if $\omega F_n = F_n$ for any permutation $\omega\in S_n$, and it is antisymmetric if $\omega F_n = (\text{sgn } \omega) F_n$, where we define $\omega (x^{\ap}) := x^{\omega \ap}=x_1^{\alpha_{\omega(1)}}\cdots x_n^{\alpha_{\omega(n)}}$.

Throughout this section, we will use the following notations:
\begin{itemize}
    \item $\de = \de(n):=(n-1,n-2,\dots,1,0) \in \bZ^n$,
    \item $F_n$: an antisymmetric formal Laurent series in $\bC[[x_1^{\pm},\dots,x_n^{\pm}]]$,
    \item $e_i=(0,\dots,0,1,0,\dots,0)$: the $i$-th standard basis of $\bZ^n$,
    \item $\be_i$: the $i$-th coordinate of $\be=(\beta_1,\dots,\beta_n) \in \bZ^n$,
    \item $\omega_t := (n-t+1,\dots,n)^{-1} \in S_n$ is a $t$-cycle and clearly $\text{sgn }\omega_t = (-1)^{t-1}$,
    \item $x:=x_1\cdots x_n$.
\end{itemize}

Now we regard $[x^{\ap}]F$ as a bilinear operation. More precisely, for $F=\sum_{\alpha} c_\alpha x^\alpha$ a polynomial in $\bC[x_1^{\pm},\dots,x_n^{\pm}]$ and $G$ a formal power series in $\bC[[x_1^{\pm},\dots,x_n^{\pm}]]$, we define
\[
[F]G = \big[\sum\limits_{\alpha} c_{\ap} x^{\ap}\big] G = \sum\limits_{\alpha} c_{\ap} [x^{\ap}] G.
\]
It is easy to see that the multiplication by $H \in \bC[x_1^{\pm},\dots,x_n^{\pm}]$ has the effect
\[
[FH]G = [F]G(x) H(x^{-1})
\]
with $H(x^{-1}) = H(x_1^{-1},\dots, x_n^{-1})$.

Given any antisymmetric formal series $F_n$,  clearly we have
\begin{align}\label{anti-sign}
[x^{\omega(\ga)}]F_{n} = (\text{sgn } \omega) [x^{\ga}] F_{n}, \quad  \forall \ga \in \bZ^{n}.
\end{align}
We denote the space of antisymmetric polynomials in $\bC[x_1^\pm,\dots,x_n^\pm]$ by  $\bC[x_1^\pm,\dots,x_n^\pm]^{\mathrm{alt}}$. We  define a linear mapping from $\bC[x_1^\pm,\dots,x_n^\pm]^{\mathrm{alt}}$  to $\mathbb{C}$ that acts as
\begin{equation}\label{phi-defn}
\varphi: \sum_{\omega \in S_n} (\text{sgn } \omega) \omega x^{\ap} \mapsto [x^{\ap}] F_n.
\end{equation}
It is easy to see that $\varphi$ is well-defined.

Recall the following fact from Macdonald's book \cite[\uppercase\expandafter{\romannumeral 3}, (1.4)]{Macdonald-book}:
\begin{align}
\sum_{\omega \in S_n} \omega \Big( \prod_{i<j}^n \frac{x_i-zx_j}{x_i-x_j} \Big) = \prod_{i=1}^n \frac{1-z^i}{1-z} =: P_n(z).
\end{align}
This can be rewritten as
\begin{align}\label{sgn-id}
\sum_{\omega \in S_n} (\text{sgn } \omega) \omega \Big( \prod_{i<j}^n (x_i-zx_j) \Big) = P_n(z) \prod_{i<j}^n (x_i-x_j).
\end{align}
In particular, if we set $z=0$, we obtain
\begin{align}\label{id-x-prod}
\sum_{\omega \in S_n} (\text{sgn } \omega) \omega(x^{\de}) = \prod_{i<j}^n (x_i-x_j).
\end{align}
 Applying $\varphi$ to both sides of \eqref{id-x-prod} multiplied by $P_n(z)$, and using \eqref{phi-defn} and \eqref{sgn-id}, we obtain
\begin{align}
\Big[\prod_{i<j}^n (x_i-zx_j)\Big] F_n = P_n(z) [x^{\de}] F_n. \label{pnz}
\end{align}

Moreover, we state a useful lemma in the work of Stembridge \cite[Lemma 3.2]{Stembridge1988}.
\begin{lem}\label{lem-series}
Let $F(z,q)$ and $G(z,q)$ be formal power series. If $F(q^k,q) = G(q^k,q)$ for infinitely many integers $k \ge 0$, then $F=G$.
\end{lem}

\subsection{Extracting terms from antisymmetric formal Laurent series }\label{sec-antisymmetric}
In this subsection, we will focus on explicit evaluations of the following way of gathering coefficients in any antisymmetric formal Laurent series $F_n \in \mathbb{C}[[x_1^\pm,x_2^\pm,\dots,x_n^\pm]]$:
\begin{align}
 W(\alpha;y):= \Big[ x^{\alpha} \prod_{i=1}^{n-1} \Big( 1-y \frac{x_n}{x_i} \Big) \Big] F_n(z,q).
\end{align}
We will evaluate $W(\alpha,y)$ for several families of exponents $\alpha$. These results are fundamental to Section \ref{sec-general} where we prove Theorems \ref{thm-SUN-N}--\ref{thm-SUN-N-1}. We have
\begin{align}\label{lem-1-expan}
    x^{\alpha} \prod_{i=1}^{n-1} \Big( 1-y \frac{x_n}{x_i} \Big) = \sum_{T \subseteq [n-1]} (-y)^t x^{\alpha}x_n^t  x^{-T} = \sum_{T \subseteq [n-1]} (-y)^t x^{\alpha+te_n} x^{-T}.
    \end{align}
Here $t:=|T|$, and when $T=\{i_1,i_2,\dots,i_t\}$ ($i_1<i_2<\cdots<i_t$) we also understand $T$ as the vector $e_{i_1}+e_{i_2}+\cdots +e_{i_t}$ so that $x^{-T} := \prod_{i \in T}x_i^{-1}$. A key fact in our proof is that the coefficient $[x^{\gamma}]F_{n}$ vanishes unless $\gamma \in\bZ^{n}$ has distinct components due to the antisymmetry of $F_n$. Thus, we will focus on finding all possible subsets $T\subseteq [n-1]$ so that the exponent vector $\alpha+te_n-T$ has distinct components.

The following lemma evaluates $W(\alpha,y)$ when $\alpha=\delta+ke_l$ where $k\in \mathbb{Z}$ and $1\leq l\leq n$.
\begin{lem}
\label{F xi k}
    Let $y$ be an indeterminate. For any $k \in \bZ$ we have
    \begin{align}
        & \Big[ x^{\de} x_l^k \prod_{i=1}^{n-1} \Big( 1-y \frac{x_n}{x_i} \Big) \Big] F_n(z,q) \notag \\
        &= \frac{1-y^{n-l}}{1-y} \big[ x^{\de} x_l^k \big] F_{n}(z,q) + \frac{y^{n-l}-y^{n}}{1-y} \big[ x^{\de} x_{l+1}^k \big] F_{n}(z,q), \quad 1 \le l \le n-1, \label{F xl k} \\
        & \Big[ x^{\de} x_{n}^k \prod_{i=1}^{n-1} \Big( 1-y \frac{x_{n}}{x_i} \Big) \big] F_{n}(z,q) = \sum_{t=0}^{n-1} y^t \big[ x^{\de} x_{n-t}^k \big] F_{n}(z,q). \label{F xn k}
    \end{align}
\end{lem}

\begin{proof}
(1)  Suppose $1\le l \le n-1$. We write the first exponent of $x$ in \eqref{lem-1-expan} as
    \begin{align}\label{lem-1-power}
    \de+ke_l+te_n = ( \underbrace{n-1,\dots,n-l+1}_{l-1},n-l+k,\underbrace{n-l-1,\dots,1}_{n-l-1},t).
    \end{align}
    To ensure that the components of this vector are distinct, $T$ must satisfy the following necessary conditions:
    \begin{enumerate}
        \item[(i)] If $i,i+1 \in [n-1] \setminus \{l\}$, then $i \in T$ implies $i+1 \in T$;
        \item[(ii)] If $i \in T$ and $i\neq l$, then $n-i-1\neq t$;
        \item[(iii)] If $i=l \in T$, then $n-l+k-1\neq t$.
    \end{enumerate}

We now divide our discussion into three cases.

\textbf{Case 1.} We first consider the case $0 \le t \le n-l-1$. If $t=0$, then $T=\emptyset$. If $1\leq t\leq n-l-1$, then $l+1\leq n-t\leq n-1$ and
    \[
    \de+ke_l+te_n = (n-1,\dots,n-l+1,\mathop{n-l+k}_{\substack{\uparrow \\ x_l}},n-l-1,\dots,\mathop{t}_{\substack{\uparrow \\ x_{n-t}}},\dots,1,t).
    \]
Since the $(n-t)$-th component and the last component are the same, we must have $n-t \in T$, and then by (i) we can get $\{ n-t,\dots,n-1\} \subseteq T$. Since $|T|=t$, $T$ must be $\{ n-t,\dots,n-1\}$. For such $T$ we have
    \begin{align*}
       &  \de+ke_l+te_n-T= (n-1,\dots,n-l+1,\mathop{n-l+k}_{\substack{\uparrow \\ x_l}},n-l-1,\dots,t+1,\mathop{t-1}_{\substack{\uparrow \\ x_{n-t}}},\dots,0,t) \\
        &= \omega_{t+1} (n-1,\dots,n-l+1,\mathop{n-l+k}_{\substack{\uparrow \\ x_l}},n-l-1,\dots,t+1,t,\mathop{t-1}_{\substack{\uparrow \\ x_{n-t+1}}},\dots,0)\\
        &= \omega_{t+1} (\de+ke_l).
    \end{align*}

\textbf{Case 2.} If $n-l \le t \le n-2$, then $1\leq n-t-1\leq l-1$ and
\[
\de + ke_l + te_n = \de+ke_l+te_n = (n-1,\dots,\mathop{t+1}_{\substack{\uparrow \\ x_{n-t-1}}},\dots,\mathop{n-l+k}_{\substack{\uparrow \\ x_l}},n-l-1,\dots,1,t).
\]
Clearly $n-t-1 \notin T$. Hence $T \subseteq \{n-t,\dots,n-1\}$. Since $|T| = t$, we must have $T = \{n-t,\dots,n-1\}$, and hence
\begin{align*}
&\de + ke_l + te_n -T = (n-1,\dots,t+1,t-1,\dots,\mathop{n-l+k-1}_{\substack{\uparrow \\ x_l}},n-l-2,\dots,0,t) \\
&= \omega_{t+1}(n-1,\dots,t+1,t,t-1,\dots,\mathop{n-l+k-1}_{\substack{\uparrow \\ x_{l+1}}},n-l-2,\dots,0) \\
&= \omega_{t+1} (\de+ke_{l+1}).
\end{align*}

\textbf{Case 3.} If $t=n-1$, then obviously we have $T=[n-1]$ and
\begin{align*}
&\de + ke_l + (n-1)e_n -T = (n-2,\dots,\mathop{n-l+k-1}_{\substack{\uparrow \\ x_l}},n-l-2,\dots,0,n-1) \\
&= \omega_{n}(n-1,n-2,\dots,\mathop{n-l+k-1}_{\substack{\uparrow \\ x_{l+1}}},n-l-2,\dots,0) \\
&= \omega_{n} (\de+ke_{l+1}).
\end{align*}
This can be combined with Case 2.

In summary, for each $0 \le t \le n-1$, there is at most one $t$-set $T \subseteq [n-1]$ such that  $\de+ke_l+te_n-T$ has distinct components. The only possible set is $T = \{n-t,\dots,n-1\}$ and for such $T$ we have
\[
    x^{\de}x_l^kx_{n}^tx^{-T} = \begin{cases}
        x^{\omega_{t+1}(\de+ke_l)}, & 0 \le t \le n-l-1 \\
        x^{\omega_{t+1}(\de+ke_{l+1})}, & n-l \le t \le n-1
    \end{cases}
    \]
with $\omega_{t+1} = (n-t,n-t+1,\dots,n)^{-1}$. Thus,
\begin{align*}
        &\bigg[ x^{\de} x_l^k \prod_{i=1}^{n-1} \Big( 1-y\frac{x_{n}}{x_i} \Big)\bigg] F_{n}
        = \bigg[ \sum_{t=0}^{n-l-1} (-y)^t x^{\omega_{t+1}(\de+ke_l)} \bigg] F_{n} + \bigg[ \sum_{t=n-l}^{n-1} (-y)^t x^{\omega_{t+1}(\de+ke_{l+1})} \bigg] F_{n} \\
       & = \frac{1-y^{n-l}}{1-y} \big[ x^{\de} x_l^k \big] F_{n} + \frac{y^{n-l}-y^{n}}{1-y} \big[ x^{\de} x_{l+1}^k \big] F_{n},
    \end{align*}
where the last equality holds due to \eqref{anti-sign}. This proves \eqref{F xl k}.

(2) Suppose $l=n$. We write the first exponent of $x$ in \eqref{lem-1-expan} as
\begin{align}\label{x-exponent-2}
\de+(k+t)e_n = (n-1,\dots,1,k+t).
\end{align}
As in part (1), to ensure that ${\de+(k+t)e_n-T}$ has distinct components, a necessary condition is that $i \in T\cap [n-2]$ implies $i+1 \in T$.
Hence, for any $0 \le t \le n-1$, there is a unique $t$-subset $\{n-t,\dots,n-1 \}$ that satisfies this condition. For such $T$ we have
\begin{align*}
&\de+(k+t)e_n-T = (n-1,\dots,\mathop{t+1}_{\substack{\uparrow \\ x_{n-t-1}}},\mathop{t-1}_{\substack{\uparrow \\ x_{n-t}}},\dots,0,k+t) \\
&= \omega_{t+1} (n-1,\dots,t+1,\mathop{k+t}_{\substack{\uparrow \\ x_{n-t}}},t-1,\dots,0)= \omega_{t+1} (\de + ke_{n-t}).
\end{align*}
It follows that
\begin{align*}
\bigg[x^{\de} x_n^k \prod_{i=1}^{n-1} \Big(1-y\frac{x_n}{x_i}\Big) \bigg] F_n = \sum_{t=0}^{n-1} (-y)^t [x^{\omega_{t+1} (\de + ke_{n-t})}] F_n = \sum_{t=0}^{n-1} y^t [x^{\de}x_{n-t}^k] F_n,
\end{align*}
which proves \eqref{F xn k}.
\end{proof}

\begin{cor}\label{cor-k01}
We have
\begin{align}
    \bigg[ x^{\de} \prod_{i=1}^{n-1} \Big( 1-y \frac{x_n}{x_i} \Big) \bigg] F_n(z,q) &= \frac{1-y^n}{1-y} [x^\de]F_n(z,q), \label{F de} \\
    \bigg[ x^{\de}x_n^{-1} \prod_{i=1}^{n-1} \Big( 1-y \frac{x_n}{x_i} \Big) \bigg] F_n(z,q) &= [x^\de x_n^{-1}]F_n(z,q), \label{add-F-de} \\
     \bigg[ x^{\de}x_n \prod_{i=1}^{n-1} \Big( 1-y \frac{x_n}{x_i} \Big) \bigg] F_n(z,q) &= y^{n-1}[x^\de x_1]F_n(z,q). \label{add-F-cor}
\end{align}
\end{cor}
\begin{proof}
Setting $k=0$ in \eqref{F xn k}, we obtain \eqref{F de}. Next, we set $k=-1$ in  \eqref{F xn k}. Note that for $1\leq t\leq n-1$, the exponents of $x_{n-t}$ and $x_{n-t+1}$ are the same in $x^\delta x_{n-t}^{-1}$  and hence $[x^\delta x_{n-t}^{-1}]F_n(z,q)=0$. This proves \eqref{add-F-de}. The identity \eqref{add-F-cor} follows in a similar way by setting $k=1$ in  \eqref{F xn k}.
\end{proof}

For $0\leq l\leq n-2$ we define
\begin{align}\label{add-beta-defn}
v_l=v_l(a,b):= (a, \underbrace{1,\dots,1}_{l}, \underbrace{0,\dots,0}_{n-l-2},b) \in \bZ^{n}.
\end{align}
We will now evaluate $W(\alpha,y)$ for $\alpha=v_l(a,b)+\delta$.

Note that
\begin{align}\label{beta-start}
v_l+\de+te_n = (a+n-1,\underbrace{n-1,\dots,n-l}_{l}, \underbrace{n-l-2,\dots,1}_{n-l-2},b+t).
\end{align}
For some special choices of $(a,b)$, we find that $W(\alpha,y)$ can be expressed as a simple combination of a few coefficients of $F_n$.
\begin{lem}\label{lem-beta-coeff}
Let $y$ be an indeterminate. For $1 \le l \le n-1$, we have
\begin{align}
\bigg[ x^{v_l(a,0)+\de}\prod_{i=1}^{n-1} \Big( 1 - y \frac{x_{n}}{x_i} \Big) \bigg] F_{n}
=\frac{1-y^{n-l-1}}{1-y} [x^{v_l(a,0)+\de}] F_n - \frac{y^{n-l-1}-y^n}{1-y} [x^{v_{l+1}(a-1,0)+\de}] F_n.
\end{align}
\end{lem}
\begin{proof}
Any subset $T$ for which $v_l(a,0)+\de + te_n-T$ has distinct components must satisfy the following conditions:
\begin{itemize}
\item[(i)] If $ i \in T$, then $i+1 \in T$ when $i \in [n-2]\setminus \{1,l+1\}$;
\item [(ii)] If $i \in T$, then $b+t \ne \begin{cases}
a+n-2, & i=1\\
n-i, & 2 \le i \le l+1\\
n-i-1, & l+2 \le i \le n-1
\end{cases}.$
\end{itemize}
 We will focus on the $(n-t)$-th component in \eqref{beta-start} which is the exponent of $x_{n-t}$.

If $t=0$, then $T=\emptyset$. If $1 \le t \le n-l-2$, then $l+2\leq n-t\leq n-1$ and
\begin{align*}
 v_l(a,0)+\de+te_n = (a+n-1,n-1,\dots,n-l, n-l-2,\dots, \mathop{t}_{\substack{\uparrow \\ x_{n-t}}},\dots,1,t).
\end{align*}
Since the $(n-t)$-th component is the same as the last component, we must have $n-t \in T$ and hence $\{n-t,\dots,n-1\} \subseteq T$. Since $|T|=t$, $T$ must be $\{n-t,\dots,n-1\}$. Therefore, we have
\begin{align*}
& v_l(a,0)+\de+te_n -T =  (a+n-1,n-1,\dots,n-l, n-l-2,\dots, t+1, \mathop{t-1}_{\substack{\uparrow \\ x_{n-t}}},\dots,0,t) \\
&= \omega_{t+1} (a+n-1,n-1,\dots,n-l, n-l-2,\dots, t+1,t,{t-1},\dots,0)\\
&= \omega_{t+1}(v_l(a,0)+\de).
\end{align*}
If $n-l-1\le t \le n-2$, then $2\leq n-t\leq l+1$ and
\begin{align*}
v_l(a,0)+\de+te_n = (a+n-1, n-1,\dots,\mathop{t+1}_{\substack{\uparrow \\ x_{n-t}}},\dots,n-l,n-l-2,\dots,1,t).
\end{align*}
Since the $(n-t)$-th component and the last component differs only by 1, we must have $n-t \notin T$, and then by (i) we deduce that $T \subseteq \{1,n-t+1,\dots,n-1\}$. Now we have $T = \{1,n-t+1,\dots,n-1\} $ since they are both $t$-sets, and
\begin{align*}
    & v_l(a,0)+\de+te_n -T \nonumber \\
    &= (a+n-2, n-1,\dots,\mathop{t+1}_{\substack{\uparrow \\ x_{n-t}}},\mathop{t-1}_{\substack{\uparrow \\ x_{n-t+1}}},\dots, n-l-1,n-l-3,\dots,0,t) \\
    &= \omega_{t} (a+n-2,n-1,\dots,t+1,t,t-1,\dots, n-l-1,n-l-3,\dots,0) \\
    &= \omega_{t} (v_{l+1}(a-1,0)+\de).
\end{align*}
For $t=n-1$, $T$ has to be $[n-1]$ and
\begin{align*}
    & v_l(a,0)+\de+(n-1)e_n -T  = (a+n-2,n-2,\dots,n-l-1,n-l-3,\dots,0,n-1)\\
    &= \omega_{n-1} (a+n-2,n-1,n-2,\dots,n-l-1,n-l-3,\dots,0) \\
   &= \omega_{n-1} (v_{l+1}(a-1,0)+\de).
\end{align*}
To summarize, for each $0 \le t \le n-1$, there is at most one $t$-set $T\subseteq[n-1]$ such that $x^{v_l(a,0)+\de}x^{-T}$ may have distinct exponents:
    \[
    T = \begin{cases}
        \{n-t,\dots,n-1\}, & 0 \le t \le n-l-2\\
        \{1,n-t+1,\dots,n-1\}, & n-l-1 \le t \le n-1
    \end{cases}.
    \]
For this set $T$ we have
    \[
    x^{v_l(a,0)+\de}x^{-T} = \begin{cases}
        x^{\omega_{t+1}(v_l+\de)}, & 0 \le t \le n-l-2\\
        x^{\omega_t(v_{l+1}(a-1,0)+\de)}, & n-l-1 \le t \le n-1
    \end{cases}.
    \]
We conclude that
    \begin{align*}
        & \bigg[ x^{v_l(a,0)+\de}\prod_{i=1}^{n-1} \Big( 1 - y \frac{x_{n}}{x_i} \Big) \bigg] F_{n} \\
        &= \sum_{t=0}^{n-l-2} (-y)^t [x^{\omega_{t+1}(v_l(a,0)+\de)}] F_n + \sum_{t=n-l-1}^{n-1} (-y)^t [x^{\omega_t(v_{l+1}(a-1,0)+\de)}] F_n \\
       &= \sum_{t=0}^{n-l-2} y^t [x^{v_l(a,0)+\de}] F_n - \sum_{t=n-l-1}^{n-1} y^t [x^{v_{l+1}(a-1,0)+\de}] F_n \\
        &= \frac{1-y^{n-l-1}}{1-y} [x^{v_l(a,0)+\de}] F_n - \frac{y^{n-l-1}-y^n}{1-y} [x^{v_{l+1}(a-1,0)+\de}] F_n,
    \end{align*}
    which is the desired result.
\end{proof}

For $r\in \mathbb{Z}$ and $0 \le s\le n-1$ we define
\begin{align}\label{beta-rs-defn}
\be(r,s) := ((r-1)(n-1)+s, \underbrace{-(r-1),\dots,-(r-1)}_{n-s-1}, \underbrace{-r,\dots,-r}_{s}) \in \bZ^{n}.
\end{align}
Notice that $\be(r,0) = \be(r-1,n-1)$ and $\be(1,0)=0$. With the help of Lemma \ref{lem-beta-coeff}, we are able to calculate $W(\alpha,y)$ for $\alpha=\beta(r,s)+\delta+ke_1$.
\begin{lem}
\label{F x be(r,s) and x1 k}
    Let $y$ be an indeterminate and $k \in \bZ$. For $1 \le s \le n-1$, we have
    \begin{align}
        & \bigg[ x^{\be(r,s)+\de}x_1^k \prod_{i=1}^{n-1} \Big( 1 - y \frac{x_{n}}{x_i} \Big) \bigg] F_{n}(z,q) \notag \\
       &= \frac{1-y^s}{1-y} \bigg[ x^{\be(r,s)+\de}x_1^k \bigg] F_{n}(z,q) - \frac{y^s-y^{n}}{1-y} \bigg[ x^{\be(r,s-1)+\de}x_1^k \bigg] F_{n}(z,q).
    \end{align}
\end{lem}
This reduces to \cite[Lemma 3.4]{Stembridge1988} when  $k=0$ and $r\ge1$ with the special series
\begin{align}
F_n(z,q)=\prod\limits_{1\leq i<j\leq n}(x_i-x_j)\prod\limits_{1\leq i\neq j\leq n} \frac{(qx_jx_i^{-1};q)_\infty}{(zx_jx_i^{-1};q)_\infty}.
\end{align}
\begin{proof}
Let $a=(r-1)(n-1)+r+k+s$ and $l=n-s-1$. Note that
\begin{align}
    \beta(r,s)+ke_1+r(e_1+\cdots+e_n)=v_l(a,0).
\end{align}
We have
\begin{align}\label{Lem2.4-proof}
    \bigg[ x^{\be(r,s)+\de}x_1^k \prod_{i=1}^{n-1} \Big( 1 - y \frac{x_{n}}{x_i} \Big) \bigg] F_{n}(z,q)=\bigg[x^{v_{l}(a,0)+\delta}\prod_{i=1}^{n-1} \Big( 1 - y \frac{x_{n}}{x_i} \Big)\bigg]x^rF_n(z,q).
\end{align}
Note that
\begin{align}\label{v-beta-relation}
v_{l+1}(a-1,0)=\beta(r,s-1)+ke_1+r(e_1+\cdots +e_n).
\end{align}
Applying Lemma \ref{lem-beta-coeff} to the right side of \eqref{Lem2.4-proof}, we obtain the desired formula.
\end{proof}

From now on we set $\ap(r):= (\underbrace{1,\dots,1}_{r},0,\dots,0) \in \bZ^{n}$ with $0 \le r \le n$. The following lemma evaluates $W(\alpha(r)+\delta,y)$ and $W(\alpha(r)+\delta+e_n,y)$.
\begin{lem}
\label{F x ap(r)}
    Let $y$ be an indeterminate. For $0 \le r \le n-1$ we have
    \begin{align}
        \label{F x ap(r)+de} \bigg[x^{\ap(r) + \de} \prod_{i=1}^{n-1} \Big( 1-y \frac{x_{n}}{x_i} \Big) \bigg] F_{n}(z,q) &= \frac{1-y^{n-r}}{1-y} \big[ x^{\ap(r)+\de} \big]F_{n}(z,q), \\
        \label{F x ap(r)+de and xn} \bigg[x^{\ap(r) + \de} x_{n} \prod_{i=1}^{n-1} \Big( 1-y \frac{x_{n}}{x_i} \Big) \bigg]F_{n}(z,q) &= \frac{y^{n-r-1}-y^{n}}{1-y} \big[ x^{\ap(r+1)+\de} \big]F_{n}(z,q).
    \end{align}
\end{lem}
\begin{proof}
(1) When $r=0$, \eqref{F x ap(r)+de} is exactly \eqref{F de}.

When $r=n-1$, we get \eqref{F x ap(r)+de} by using \eqref{add-F-de} and noting that
\begin{align}\label{lem-2.6-proof}
\bigg[x^{\alpha(n-1)+\de}\prod_{i=1}^{n-1} \Big( 1-y \frac{x_{n}}{x_i}\ \Big)\bigg]F_n=\bigg[x^{\delta}x_n^{-1}\prod_{i=1}^{n-1} \Big( 1-y \frac{x_{n}}{x_i}\ \Big)\bigg] x^{-1}F_n.
\end{align}

When $1\leq r \leq n-2$ we have $\alpha(r)=v_{r-1}(1,0)$. Applying Lemma \ref{lem-beta-coeff} with $l=r-1$ we obtain  \eqref{F x ap(r)+de} using the simple fact that
\begin{align}
[x^{v_{r}(0,0)+\delta}]F_n=0, \quad 1\leq r \leq n-2.
\end{align}

 (2) Now we turn to prove  \eqref{F x ap(r)+de and xn}. When $r=0$, the formula  \eqref{F x ap(r)+de and xn} is exactly \eqref{add-F-cor}. When $r=n-1$, clearly by \eqref{F de} we have
    \[
    \bigg[x^{\ap(n-1) + \de} x_{n} \prod_{i=1}^{n-1} \Big( 1-y \frac{x_{n}}{x_i} \Big) \bigg]F_{n} = \bigg[x^{\de} \prod_{i=1}^{n-1} \Big( 1-y \frac{x_{n}}{x_i} \Big) \bigg] x^{-1}F_{n}=\frac{1-y^n}{1-y} [x^\de x^{\alpha(n)}]F_n.
    \]

    It remains to show \eqref{F x ap(r)+de and xn} for $1 \le r \le n-2$. Note that $\alpha(r)+e_n=v_{r-1}(1,1)$. We denote $l=r-1$ and note that $0 \le l \le n-3$.

    If $0 \le t \le n-l-3$, then $l+2\leq n-t-1\leq n-1$ and we write \eqref{beta-start} as
    \[
    v_l(1,1)+\de = (n,n-1,\dots,n-l,n-l-2,\dots,\mathop{t+1}_{\substack{\uparrow \\ x_{n-t-1}}},\dots,1,t+1).
    \]
It follows that $n-t-1 \in T$ and hence $\{n-t-1,\dots,n-1\} \subseteq T$. However, the preceding set is a $(t+1)$-set, contradicting to the fact $|T| = t$. Thus,
\begin{align}
[x^{v_l(1,1)+\de+|T|e_n}x^{-T}]F_n = 0, \quad 0 \le t \le n-l-3.
\end{align}

If $n-l-2\leq t\leq n-1$, then $0\leq n-t-1\leq l+1$ and we have
\begin{align}\label{case-beta-1}
    v_l(1,1)+\de+te_n = (n,n-1,\dots, \mathop{t+2}_{\substack{\uparrow \\ x_{n-t-1}}},t+1,\dots, n-l,n-l-2,\dots,1,t+1).
\end{align}
    Clearly, in this case $n-t-1 \notin T$ and hence $T \subseteq \{n-t,\dots, n-1\}$. Since $|T|=t$ we must have $T = \{n-t,\dots, n-1\}$. Therefore,
    \begin{align*}
       & v_l(1,1)+\de+te_n -T = (n,n-1,\dots, \mathop{t+2}_{\substack{\uparrow \\ x_{n-t-1}}},t,\dots, n-l-1,n-l-3,\dots,0,t+1) \\
        &= \omega_{t+1} (n,n-1,\dots,t+2,t+1,t,\dots,n-l-1,n-l-3,\dots,0) \\
        &= \omega_{t+1} (v_{l+1}(1,0)+\de).
    \end{align*}
Combining the above cases, we deduce that
    \begin{align*}
        & \bigg[x^{v_l(1,1) + \de}  \prod_{i=1}^{n-1} \Big( 1-y \frac{x_{n}}{x_i} \Big) \bigg]F_{n}(z,q) = \sum_{t=n-l-2}^{n-1} (-y)^t [x^{\omega_{t+1} (v_{l+1}(1,0)+\de)}] F_n(z,q) \\
        &= \sum_{t=n-l-2}^{n-1} y^t [x^{v_{l+1}(1,0)+\de}] F_n(z,q) = \frac{y^{n-l-2}-y^n}{1-y} [x^{v_{l+1}(1,0)+\de}] F_n(z,q). \qedhere
    \end{align*}
\end{proof}

Replacing $F_n$ by $xF_n$ and $r$ by $r-1$ in Lemma \ref{F x ap(r)}, we obtain the following byproduct.
\begin{cor}
\label{F x -ap(r)}
For $1 \le r \le n$ we have
    \begin{align}
        \label{F x de-ap(r) xn -1} \bigg[ \frac{x^\de}{x_{r}\cdots x_{n}}  \prod_{i=1}^{n-1} \Big( 1-y \frac{x_{n}}{x_i} \Big) \bigg] F_n &= \frac{1-y^{n-r+1}}{1-y} \bigg[ \frac{x^\de}{x_{r}\cdots x_{n}} \bigg]F_n, \\
        \label{F x de-ap(r)} \bigg[ \frac{x^\de}{x_{r}\cdots x_{n-1}} \prod_{i=1}^{n-1} \Big( 1-y \frac{x_{n}}{x_i} \Big) \bigg]F_n &= \frac{y^{n-r}-y^{n}}{1-y} \bigg[ \frac{x^\de}{x_{r+1}\cdots x_{n}} \bigg]F_n.
    \end{align}
\end{cor}

Finally, we give a formula for $W(\delta+ke_1-e_n,y)$, i.e., $W(\delta+v_0(k,-1),y)$.
\begin{lem}
\label{F x1 k}
    Let $y$ be an indeterminate. For any $k\in\bZ$ we have
    \begin{equation}
        \bigg[x^{\de}x_1^kx_{n}^{-1} \prod_{i=1}^{n-1} \Big( 1-y \frac{x_{n}}{x_i} \Big) \bigg] F_{n}(z,q) = \big[ x^{\de}x_1^kx_{n}^{-1} \big]F_{n}(z,q) - \frac{y-y^{n}}{1-y} \big[ x^{\de}x_1^{k-1} \big]F_{n}(z,q). \label{F x de x1 k and xn -1}
    \end{equation}
\end{lem}
\begin{proof}
We write the exponent vector as
\begin{align}
    \de+ke_1+(t-1)e_n = (n-1+k,n-2,\dots,1,t-1).
\end{align}

    When $t=0$, we have $T = \emptyset$.

    If $1 \le t \le n-2$, then $2\leq n-t \leq n-1$ and we have
    \[
    \de+ke_1+(t-1)e_n = (n-1+k,n-2,\dots,\mathop{t}_{\substack{\uparrow \\ x_{n-t}}},\dots,1,t-1).
    \]
It follows that $n-t \notin T$ and hence $T \subseteq\{1,n-t+1,\dots,n-1\}$. Since $|T|=t$ we must have $T =\{1,n-t+1,\dots,n-1\}$. Therefore,
    \begin{align*}
      &  \de+ke_1+(t-1)e_n -T = (n-2+k,n-2,\dots,t,t-2,\dots,0,t-1) \\
        &= \omega_t (n-2+k,n-2,\dots,t,t-1,t-2,\dots,0) \\
        &= \omega_t (\de+(k-1)e_1).
    \end{align*}

 When $t=n-1$, $T$ has to be $[n-1]$ and
    \begin{align*}
        &\de+ke_1+(t-1)e_n -T = (n-2+k,n-3,\dots,0,n-2) \\
        &= \omega_{n-1} (n-2+k,n-2,\dots,0) = \omega_{n-1} (\de+(k-1)e_1).
    \end{align*}
 Combining the above cases, we deduce that
    \begin{align*}
        & \bigg[x^{\de}x_1^kx_{n}^{-1} \prod_{i=1}^{n-1} \Big( 1-y \frac{x_{n}}{x_i} \Big) \bigg] F_n = [x^{\de}x_1^kx_n^{-1}]F_n + \sum_{t=1}^{n-1} (-y)^t [x^{\omega_t(\de+(k-1)e_1)}]F_n \\
        &= [x^{\de}x_1^kx_n^{-1}]F_n - \frac{y-y^n}{1-y} [x^\de x_1^{k-1}]F_n.
    \end{align*}
This proves the desired assertion.
\end{proof}


\section{Proofs and generalizations of the ${\rm SU}(N)$ conjectures}\label{sec-general}
We now present proofs and generalizations of the conjectures mentioned in Section \ref{sec-intro}. We will rely on the crucial results established in Section \ref{sec-pre}.

\subsection{The ${\rm SU}(N)_{-N}$ conjecture}\label{sec-N}
We establish the following constant term identity which generalizes Theorem \ref{thm-SUN-N}.
\begin{thm}\label{thm-CT-1}
Let $x_1x_2\cdots x_n=1$ and
\begin{align}\label{eq-thm-CT-0}
    C_k(z):=[(x_1x_2\cdots x_{n-1})^0]  F_{k,n}(x_1,x_2,\dots,x_{n-1},(x_1x_2\cdots x_{n-1})^{-1})
\end{align}
where
\begin{align}
F_{k,n}(x):= \prod_{i \ne j}^{n} \frac{(x_ix_j^{-1};q)_\infty}{(zx_ix_j^{-1};q)_\infty} \sum_{i=1}^n x_i^{k} \in \bZ[[x_1^{\pm},\dots,x_{n}^{\pm}]].
\end{align}
We have
\begin{align}\label{eq-thm-CT-zero}
C_k(z)=0, \quad k\not\equiv 0 \pmod{n}
\end{align}
and
\begin{align}\label{eq-thm-CT-1}
   & C_{nm}(z)= n!(-1)^{(n-1)m}q^{(n-1)m(m-1)/2} \bigg( 1+ \sum_{l=1}^{n-1} \prod_{i=1}^l \frac{(q^m-z^{i})(1-z^{n-i})}{(1-q^{m}z^{n-i})(1-z^{i})} \bigg) \nonumber \\
   &\qquad \qquad \times \frac{(z;q)_{\ii}^{n-1}}{(q;q)_{\ii}^{n-2} (qz^{n-1};q)_{\ii}(z;z)_{n-1}}  \prod_{i=1}^{n-1} \frac{(z^iq^{1-m},z^iq^m;q)_{\ii}}{(z^{i+1},z^{i-1}q;q)_{\ii}}.
\end{align}
\end{thm}
We have
\begin{align}
C_k(z)= \sum_{m=-\ii}^{\ii} [x^m] F_{k,n} (x).
\end{align}
Since $(x_ix_j^{-1};q)_\infty/(zx_ix_j^{-1};q)_\infty$ is homogeneous of degree 0,  $[x^m] F_{k,n} (x)$ vanishes unless $k=mn$ since the degree of each monomial in $F_{k,n}$ is $k$ and that of $x^m$ is $mn$. Hence
\begin{align}\label{F-vanish}
   C_k(z) =0, \quad k\not\equiv 0 \pmod{n}
\end{align}
and
\begin{align}\label{constant-f-F}
C_{mn}(z) = [x^m]F_{mn,n}(x).
\end{align}
Now we connect $[x^m]F_{mn,n}(x)$ with an antisymmetric series so that the results in Section \ref{sec-antisymmetric} can be employed. We have
\begin{align*}
    & [x^m]F_{mn,n}(x) = [x^0] x^{-m} \prod_{i \ne j}^{n} \frac{(x_ix_j^{-1};q)_\infty}{(zx_ix_j^{-1};q)_\infty}  \sum_{i=1}^n x_i^{nm} \\
    &= [x^0] \prod_{i<j}^n (1-x_ix_j^{-1})(1-x_jx_i^{-1}) \cdot x^{-m} \prod_{i \ne j}^{n} \frac{(qx_ix_j^{-1};q)_{\infty}}{(zx_ix_j^{-1};q)_\infty}  \sum_{i=1}^n x_i^{nm} \\
    &= [x^0] \prod_{i<j}^n (x_i-x_j)(x_i^{-1}-x_j^{-1}) \cdot x^{-m} \prod_{i \ne j}^{n} \frac{(qx_ix_j^{-1};q)_{\infty}}{(zx_ix_j^{-1};q)_\infty}  \sum_{i=1}^n x_i^{nm} \\
    &= \bigg[ \prod_{i<j}^n (x_i-x_j) \bigg] g_{m,n}(z,q)\sum_{i=1}^n x_i^{nm}
\end{align*}
where
\begin{align}\label{add-g-defn}
    g_{m,n}(z,q) := x^{-m} \prod_{i \ne j}^{n} \frac{(qx_ix_j^{-1};q)_{\ii}}{(zx_ix_j^{-1};q)_{\ii}} \prod_{i<j}^{n} (x_i-x_j).
\end{align}
Since $g_{m,n}(z,q)\sum_{i=1}^n x_i^{nm}$ is antisymmetric for $x_1,\dots,x_n$, by \eqref{pnz} we have
\begin{align}\label{add-F-start}
 [x^m]F_{mn,n}(x) = n! [x^\de] g_{m,n}(z,q) \sum_{i=1}^n x_i^{nm}.
\end{align}

We denote
\begin{align}\label{A-defn}
A_i(z,q) = A_{i,m,n}(z,q):= [x^\de x_i^{-nm}] g_{m,n}
\end{align}
and sometimes abbreviate $A_i(z,q)$ as $A_i$.
Then from \eqref{add-F-start} we have
\begin{align}\label{add-F-decompose}
[x^m] F_{mn,n}(x)  = n! \sum_{i=1}^n A_i(z,q).
\end{align}
We now aim to find relations between $A_i$ and compute them explicitly. When $m=0$, it is easy to see that
\begin{align}
    A_1=A_2= \cdots = A_{n}.
\end{align}
The general case is more complicated and we will discuss as follows.

We first state a lemma which generalizes \cite[Lemma 3.3]{Stembridge1988}.
\begin{lem}
\label{qx for gm,n}
    For any $\be \in \bZ^{n}$, we have
    \begin{align}\label{eq-lem-gmn}
    q^{\be_{n}+m} \bigg[ x^{\be} \prod_{i=1}^{n-1} \Big(1-z\frac{x_{n}}{x_i}\Big) \bigg] g_{m,n}(z,q) = z^{n-1} \bigg[ x^{\be} \prod_{i=1}^{n-1} \Big(1-z^{-1} \frac{x_{n}}{x_i}\Big) \bigg] g_{m,n}(z,q).
    \end{align}
\end{lem}
This corresponds to the special case $y=0$ of Lemma \ref{qx for gk} which we will prove later.

Now we are able to figure out the relation between the $A_i$'s.
\begin{cor}
For $1 \le l \le n-1$ we have
\begin{align}\label{add-Cor2.9-relation}
    A_{l+1} = \prod_{i=1}^l\frac{(q^m-z^i)(1-z^{n-i})}{(1-q^mz^{n-i})(1-z^i)} A_1.
\end{align}
\end{cor}

\begin{proof}
From Lemma \ref{qx for gm,n} with $\be = \de-nme_l$ we deduce that
\begin{align}\label{add-Cor2.9-start}
    q^m \bigg[x^{\de}x_l^{-nm} \prod_{i=1}^{n-1} \Big(1-z\frac{x_n}{x_i}\Big) \bigg] g_{m,n} =z^{n-1} \bigg[ x^{\de} x_l^{-nm} \prod_{i=1}^{n-1} \Big(1-z^{-1} \frac{x_{n}}{x_i}\Big) \bigg] g_{m,n}.
\end{align}
Applying \eqref{F xl k} with $y=z$ to the left side of \eqref{add-Cor2.9-start}, we have
    \begin{align}\label{Cor2.9-LHS}
        & q^m \bigg[x^{\de}x_l^{-nm} \prod_{i=1}^{n-1} \Big(1-z\frac{x_n}{x_i}\Big) \bigg] g_{m,n}  \nonumber \\
        &= q^m \frac{1-z^{n-l}}{1-z} [x^\de x_l^{-nm}] g_{m,n} + q^m \frac{z^{n-l}-z^n}{1-z} [x^\de x_{l+1}^{-nm}] g_{m,n} \nonumber \\
        &= q^m \frac{1-z^{n-l}}{1-z} A_l + q^m \frac{z^{n-l}-z^n}{1-z} A_{l+1}.
    \end{align}
Applying \eqref{F xl k} with $y=z^{-1}$ to the right side of \eqref{add-Cor2.9-start}, we have
    \begin{align}\label{Cor2.9-RHS}
        & z^{n-1} \bigg[ x^{\de} x_l^{-nm} \prod_{i=1}^{n-1} \Big(1-z^{-1} \frac{x_{n}}{x_i}\Big) \bigg] g_{m,n}  \nonumber \\
        &= z^{n-1} \frac{1-z^{l-n}}{1-z^{-1}} [x^\de x_l^{-nm}] g_{m,n} + z^{n-1} \frac{z^{l-n}-z^{-n}}{1-z^{-1}} [x^\de x_{l+1}^{-nm}] g_{m,n}  \nonumber \\
        &= \frac{z^l-z^n}{1-z} A_l + \frac{1-z^l}{1-z} A_{l+1}.
    \end{align}
Substituting \eqref{Cor2.9-LHS} and \eqref{Cor2.9-RHS} into \eqref{add-Cor2.9-start}, we deduce that
\begin{align}\label{add-A-relation}
    A_{l+1} = \frac{(q^m-z^l)(1-z^{n-l})}{(1-q^mz^{n-l})(1-z^l)}A_l, \quad 1\le l \le  n-1.
\end{align}
Iterating it we obtain \eqref{add-Cor2.9-relation}.
\end{proof}

Similar to the previous corollary, if we apply \eqref{F xn k} to Lemma \ref{qx for gm,n} with $\be = \de-nme_{n}$, then we deduce that
    \[
    \sum_{l=1}^{n} (q^{-(n-1)m}z^{n-l}-z^{l-1})A_l = 0.
    \]
If we write them in term of $A_1$, then we  obtain a curious identity about $z$ which will not be needed in this paper but is of independent interest itself.
\begin{cor}
For any $n \in \bN^+$ and $m \in \bZ$ we have
\begin{align}\label{z-id}
    \sum_{l=0}^n (q^{-nm}z^{n-l}-z^l) \prod_{i=1}^{l}\frac{(q^m-z^i)(1-z^{n-i+1})}{(1-q^mz^{n-i+1})(1-z^i)} =0.
\end{align}
\end{cor}
\begin{rem}
This can also be proved directly using the $q$-Chu-Vandermonde summation formula \cite[(\uppercase\expandafter{\romannumeral 2}.6)]{GR-book}
\begin{align}\label{q-Chu}
    &{}_2\phi_1\bigg(\genfrac{}{}{0pt}{}{a,q^{-n}}{c};q,q\bigg)=\frac{(c/a;q)_n}{(c;q)_n}a^n
\end{align}
and its invariant (reversing the order of summation) \cite[(\uppercase\expandafter{\romannumeral 2}.7)]{GR-book}
\begin{align}\label{q-Chu-variant}
    &{}_2\phi_1\bigg(\genfrac{}{}{0pt}{}{a,q^{-n}}{c};q,cq^n/a\bigg)=\frac{(c/a;q)_n}{(c;q)_n}.
\end{align}
Here the ${}_r\phi_s$ series is defined as (see e.g.~\cite[(1.2.22)]{GR-book})
\begin{align}
     &{}_r\phi_s\bigg(\genfrac{}{}{0pt}{}{a_1,a_2,\dots,a_r}{b_1,b_2,\dots,b_s};q,z\bigg)=\sum_{n=0}^\infty \frac{(a_1,a_2,\dots,a_r;q)_n}{(q,b_1,b_2,\dots,b_s;q)_n}\Big((-1)^nq^{n(n-1)/2}\Big)^{1+s-r}z^n
\end{align}
where $q\neq 0$ when $r>s+1$.
In fact, if we write $q^m=t$, then the left side of \eqref{z-id} becomes
\begin{align}\label{z-id-proof}
  & t^{-n}z^n {}_2\phi_1\bigg(\genfrac{}{}{0pt}{}{t^{-1}z,z^{-n}}{t^{-1}z^{-n}};z,z^{-1}\bigg)-{}_2\phi_1\bigg(\genfrac{}{}{0pt}{}{t^{-1}z,z^{-n}}{t^{-1}z^{-n}};z,z\bigg)\nonumber \\
   &=t^{-n}z^n\frac{(z^{-n-1};z)_n}{(t^{-1}z^{-n};z)_n}-\frac{(z^{-n-1};z)_n}{(t^{-1}z^{-n};q)_n}(t^{-1}z)^n=0.
\end{align}
Here the first equality follows from \eqref{q-Chu} and \eqref{q-Chu-variant}. This proves \eqref{z-id}.
\end{rem}

Therefore, we have
\begin{align}
    \label{f and A1}
    &[x_1^0x_2^0\cdots x_{n-1}^0] F_{mn,n}(x_1,\dots,x_{n-1},(x_1x_2\cdots x_{n-1})^{-1}) = n! \sum_{l=1}^{n} A_l(z,q)  \nonumber \\
    &= n! \bigg( 1+ \sum_{l=1}^{n-1} \prod_{i=1}^l \frac{(q^m-z^{i})(1-z^{n-i})}{(1-q^{m}z^{n-i})(1-z^{i})} \bigg) A_1(z,q).
\end{align}
When $m=0$, from  Stembridge's result \cite[Theorem 3.1]{Stembridge1988} we have
\begin{equation}
\label{g0,n}
   A_1= [x^{\de}] g_{0,n}(z,q) =\frac{ (z;q)_{\ii}^{n}}{(q;q)_{\ii}^{n-1} (z^{n};q)_{\ii} (z;z)_{n-1}}.
\end{equation}

We will give a computation of $A_1$ for general $m$ following the method in Stembridge's proof of \cite[Theorem 3.1]{Stembridge1988}. This will be a key result in the proof of our theorems.
\begin{thm}
    \label{A1}
    We have
    \begin{align}\label{eq-thmA1}
    A_{1,m,n}(z,q) &= (-1)^{(n-1)m}q^{(n-1)m(m-1)/2} \frac{(z;q)_{\ii}^{n-1}}{(q;q)_{\ii}^{n-2} (qz^{n-1};q)_{\ii}(z;z)_{n-1}} \nonumber \\
    &\qquad \times \prod_{i=1}^{n-1} \frac{(z^iq^{1-m},z^iq^m;q)_{\ii}}{(z^{i+1},z^{i-1}q;q)_{\ii}}.
    \end{align}
\end{thm}
\begin{proof}
Setting $\beta=\beta(r,s)+\delta-nme_1$ in Lemma \ref{qx for gm,n} with $s>0$, we obtain
\begin{align}\label{proof-Thm-A1-start}
&q^{m-r} \bigg[ x^{\beta(r,s)+\delta} x_1^{-nm}\prod_{i=1}^{n-1} \Big(1-z\frac{x_{n}}{x_i}\Big) \bigg] g_{m,n}(z,q) \nonumber \\
&= z^{n-1} \bigg[ x^{\beta(r,s)+\delta}x_1^{-nm} \prod_{i=1}^{n-1} \Big(1-z^{-1} \frac{x_{n}}{x_i}\Big) \bigg] g_{m,n}(z,q).
\end{align}
Applying Lemma \ref{F x be(r,s) and x1 k} with $k=-nm$ and $y=z$ (resp.~$y=z^{-1}$), we obtain equivalent expressions for the left (resp.~ right) side of \eqref{proof-Thm-A1-start}, then after simplification we obtain the following relation:
    \begin{align}\label{Thm2.11-rec}
    [x^{\be(r,s)+\de}x_1^{-nm}]g_{m,n} = \frac{(q^mz^s-q^r)(1-z^{n-s})}{(q^m-q^rz^{n-s})(1-z^s)} [x^{\be(r,s-1)+\de}x_1^{-nm}]g_{m,n}.
    \end{align}
    Recalling the fact that $\be(r+1,0) = \be(r,n-1)$, using \eqref{Thm2.11-rec} to iterate we have
    \begin{align}\label{Thm.11-rec-2}
        &[x^{\be(r+1,0)+\de}x_1^{-nm}]g_{m,n}=[x^{\beta(r,n-1)+\delta}x_1^{-nm}]g_{m,n}  \nonumber\\
        &= \prod_{s=1}^{n-1} \frac{(q^mz^s-q^r)}{(q^m-q^rz^{n-s})} [x^{\be(r,0)+\de}x_1^{-nm}]g_{m,n} \nonumber \\
        &= \prod_{s=1}^{n-1} \frac{(q^mz^s-q^r)}{(q^m-q^rz^s)} [x^{\be(r,0)+\de}x_1^{-nm}]g_{m,n}.
    \end{align}
Recall that $A_1=[x^{\beta(1,0)+\delta}x_1^{-nm}]g_{m,n}$ since $\beta(1,0)=0$. Iterating \eqref{Thm.11-rec-2} we deduce that for any $k\in\bZ_{\ge0}$
    \[
    [x^{\be(k,0)+\de}x_1^{-nm}]g_{m,n} = \prod_{r=1}^{k-1} \prod_{s=1}^{n-1} \frac{(q^mz^s-q^r)}{(q^m-q^rz^s)} A_1.
    \]
In other words, we have
\begin{align}\label{Thm2.11-A1}
    A_1(z,q) = \prod_{r=1}^{k-1} \prod_{s=1}^{n-1} \frac{(q^m-q^rz^s)}{(q^mz^s-q^r)} [x^{\be(k,0)+\de}x_1^{-nm}]g_{m,n}(z,q).
\end{align}
    Now we choose $k> m$ and set $z=q^{k-m}$, then we have
    \begin{align}\label{add-coeff-g}
        &[x^{\be(k,0)+\de}x_1^{-nm}]g_{m,n}(q^{k-m},q)  \nonumber \\
       &= [x_1^{(k-m-1)(n-1)}x_2^{-(k-m-1)} \cdots x_{n}^{-(k-m-1)}] \prod_{i<j}^{n} (qx_ix_j^{-1};q)_{k-m-1} (x_jx_i^{-1};q)_{k-m} \nonumber \\
        &=(-1)^{(n-1)(k-m-1)} q^{(n-1)(k-m-1)(k-m)/2}[x^{\de(n-1)}] g_{0,n-1}(q^{k-m},q).
    \end{align}
Here the last equality follows from \cite[(6)--(7), pp.~782--783]{Stembridge1988}.

Note that
    \begin{align}\label{add-A1-exp}
         &\prod_{r=1}^{k-1} \prod_{s=1}^{n-1} \frac{(q^m-q^rq^{s(k-m)})}{(q^mq^{s(k-m)}-q^r)}= \prod_{s=1}^{n-1} \prod_{r=1}^{k-1} (-q^{m-r}) \frac{1-q^{s(k-m)+r-m}}{1-q^{s(k-m)+m-r}}   \\
        &= (-1)^{(n-1)(k-1)} q^{(k-1)(n-1)m-(n-1)k(k-1)/2} \nonumber \\
        &\qquad \times \prod_{s=1}^{n-1} \frac{(1-q^{s(k-m)+1-m}) \cdots (1-q^{s(k-m)+k-m-1})}{(1-q^{s(k-m)+m-1}) \cdots (1-q^{s(k-m)+m-k+1})}  \nonumber  \\
        &= (-1)^{(n-1)(k-1)} q^{(k-1)(n-1)m-(n-1)k(k-1)/2}  \nonumber \\
        &\qquad \times \prod_{s=1}^{n-1} \frac{(1-q^{s(k-m)+1-m}) \cdots (1-q^{(s+1)(k-m)-1})}{(1-q^{(s-1)(k-m)+1}) \cdots (1-q^{s(k-m)+m-1})}  \nonumber  \\
        &= (-1)^{(n-1)(k-1)} q^{(k-1)(n-1)m-(n-1)k(k-1)/2} \prod_{s=1}^{n-1} \frac{(q^{s(k-m)+1-m};q)_{\ii} (q^{s(k-m)+m};q)_{\ii}}{(q^{(s+1)(k-m)};q)_{\ii} (q^{(s-1)(k-m)+1};q)_{\ii}}. \nonumber
    \end{align}
 Substituting \eqref{add-coeff-g} and \eqref{add-A1-exp} into  \eqref{Thm2.11-A1} and utilizing \eqref{g0,n}, we deduce that
\begin{align}
   & A_1(q^{k-m},q)
    =(-1)^{(n-1)m}q^{(n-1)m(m-1)/2} \prod_{s=1}^{n-1} \frac{(q^{s(k-m)+1-m};q)_{\ii} (q^{s(k-m)+m};q)_{\ii}}{(q^{(s+1)(k-m)};q)_{\ii} (q^{(s-1)(k-m)+1};q)_{\ii}}  \nonumber \\
    &\qquad \qquad \qquad \times \frac{(q^{k-m};q)_{\ii}^{n-1}}{(q;q)_{\ii}^{n-2}(q^{1+(k-m)(n-1)};q)_{\ii}(q^{k-m};q^{k-m})_{n-1}}.
\end{align}
This proves the desired formula \eqref{eq-thmA1} for $z=q^{k-m}$.
Since $k>m$ is arbitrary, by Lemma \ref{lem-series} we see that \eqref{eq-thmA1} holds for any $z$ with $|z|<1$.
\end{proof}

\begin{proof}[Proof of Theorem \ref{thm-CT-1}]
From \eqref{F-vanish} we get \eqref{eq-thm-CT-zero}. Substituting \eqref{eq-thmA1} into \eqref{f and A1}, we obtain \eqref{eq-thm-CT-1}.
\end{proof}

\begin{proof}[Proof of Theorem \ref{thm-SUN-N}]
This follows from Theorem \ref{thm-CT-1} with $z=0$.
\end{proof}

\subsection{The ${\rm SU}(N)_{-N-\frac12}$ conjecture}\label{sec-N12}
Our first goal is to establish the following constant term identity that generalizes Theorem \ref{thm-SUN-N-12}.
\begin{thm}\label{thm-CT-SU-N12}
Let
\begin{align}
    C_m^{(1)}(y,z):=[(x_1x_2\cdots x_{n-1})^0] F(x_1,x_2,\dots,x_{n-1},(x_1x_2\dots x_{n-1})^{-1})
\end{align}
where
\begin{align}
 F_{m,n}^{(1)}(x) := \prod_{i\ne j}^n \frac{(x_ix_j^{-1};q)_\infty}{(x_ix_j^{-1}z;q)_\infty} \prod_{i=1}^n \frac{(q^{\frac12}x_i y;q)_\infty}{(q^{\frac12}x_i yz;q)_\infty} \sum_{i=1}^n x_i^m \in \bZ[x_1^{\pm},\dots,x_n^{\pm}]. \label{def-Fmn-1}
\end{align}
We have
\begin{align}
   C_0^{(1)}(y,z)&=\frac{n n! (z;q)_{\ii}^{n}}{(q)_{\ii}^{n-1} (z^{n};q)_{\ii} (z;z)_{n-1}} \sum_{k=0}^\infty (yq^{\frac12})^{nk} \prod_{t=1}^k \prod_{r=1}^{n}\frac{z^r-q^{t-1}}{1-q^tz^{r-1}}, \label{eq-thm-SU-N12-C0}\\
   C_1^{(1)}(y,z)&=\frac{n! (1-z^n)(z;q)_{\ii}^{n}}{(q;q)_{\ii}^{n-1} (1-z)(z^{n};q)_{\ii} (z;z)_{n-1}} \sum_{k=1}^\infty \frac{(yq^{\frac12})^{nk-1}(1-q^k)}{z^n-q^{k-1}} \prod_{t=1}^k \prod_{r=1}^{n}\frac{z^r-q^{t-1}}{1-q^tz^{r-1}}. \label{eq-thm-SU-N12-C1}
\end{align}
\end{thm}
Note that we can also write \eqref{eq-thm-SU-N12-C0} as
\begin{align}
     C_0^{(1)}(y,z)= \frac{n n! (z;q)_{\ii}^{n}}{(q)_{\ii}^{n-1} (z^{n};q)_{\ii} (z;z)_{n-1}} {}_{n+1}\phi_n\bigg(\genfrac{}{}{0pt}{}{z^{-1},z^{-2},\dots,z^{-n}}{qz,qz^2,\dots,qz^{n-1}};q,(yq^{1/2})^nz^{n(n+1)/2}\bigg).
\end{align}

We now proceed to prove Theorem \ref{thm-CT-SU-N12}. We have
\begin{align}\label{General-2-start}
C_m^{(1)}(y,z)= \sum_{k=-\ii}^{\ii} [x^k] F_{m,n}^{(1)}(x).
\end{align}
Note that the infinite product $(x_ix_j^{-1};q)_\infty/(x_ix_j^{-1}z;q)_\infty$ is homogeneous of degree $0$, and $(q^{\frac12}x_iy;q)_\infty/(q^{\frac12}x_i yz;q)_\infty$ contains only monomials with nonnegative powers. Therefore, the coefficient $[x^k]F_{m,n}^{(1)}(x)$ vanishes when $nk<m$. Hence \eqref{General-2-start} can be rewritten as
\begin{align}\label{General-2-key}
C_m^{(1)}(y,z)  = \sum_{k\ge \lceil m/n \rceil}^{\ii} [x^k] F_{m,n}^{(1)}(x),
\end{align}
where $\lceil m/n \rceil$ is denoted as the minimal integer that is no less than $m/n$.
Similarly to Section \ref{sec-N}, we have
\begin{align}\label{add-coeff}
[x^k] F_{m,n}^{(1)}(x) = n! [x^{\de}] h_k(z,q) \sum_{i=1}^n x_i^m
\end{align}
where
\begin{align}\label{hk-defn}
h_k(z,q) := x^{-k} \prod_{i<j}^n (x_i-x_j) \prod_{i\ne j}^n \frac{(qx_ix_j^{-1};q)_{\ii}}{(zx_ix_j^{-1};q)_{\ii}} \prod_{i=1}^n \frac{(q^{\frac12}x_iy;q)_{\ii}}{(q^{\frac12}x_iyz;q)_{\ii}}.
\end{align}
It is easy to see that $h_k(z,q) \cdot \sum_{i=1}^n x_i^m$ is antisymmetric for $x_1,\dots,x_n$.

\begin{lem}
\label{qx for gk}
    For any $\be \in \bZ^n$ and $1\leq \ell \leq n$, we have
    \begin{align}\label{add-h-relation}
        &q^{\be_\ell+k} \bigg[x^{\be} (1-q^{-\frac12}yx_\ell^{-1}) \prod_{\substack{i=1 \\ i\neq \ell}}^{n} \Big(1-z\frac{x_\ell}{x_i}\Big) \bigg] h_k(z,q) \nonumber  \\
        &= z^{n-1} \bigg[x^{\be} (1-q^{\frac12}zyx_\ell^{-1}) \prod_{\substack{i=1 \\ i\neq \ell}}^{n} \Big(1-z^{-1}\frac{x_\ell}{x_i}\Big) \bigg] h_k(z,q).
    \end{align}
\end{lem}
When $y=0$, $h_k(z;q)$ becomes $g_{k,n}(z,q)$ and this lemma reduces to Lemma \ref{qx for gm,n}. Before giving a proof, we state an easy but useful fact: for any $f(x) \in \bC[x_1^{\pm},\dots,x_n^{\pm}]$ we have
\begin{equation}
    \label{fact of qx}
q^{\be_i} [x^\be] f(x_1,\dots,x_n) = [x^{\be}] f(x_1,\dots,x_{i-1},qx_i,x_{i+1},\dots,x_n).
\end{equation}

\begin{proof}
     Through a straightforward computation, we find that
    \[
    \frac{h_{k}(z,q;x_{\ell} \to qx_{\ell})}{h_k(z,q)} = z^{n-1}q^{-k} \prod_{\substack{i=1 \\ i\neq \ell}}^{n} \frac{(1-z^{-1}x_ix_{\ell}^{-1})}{(1-zq^{-1}x_ix_{\ell}^{-1})} \times \frac{1-q^{\frac{1}{2}}x_\ell yz}{1-q^{\frac{1}{2}}x_\ell y}.
    \]
    In other words, we have
    \begin{align}\label{proof-h-relation}
   & h_k(z,q) \prod_{\substack{i=1\\ i\neq \ell}}^{n} (1-zx_ix_{\ell}^{-1})(1-q^{-\frac{1}{2}}x_\ell y) \bigg|_{x_\ell \to qx_\ell}  \nonumber \\
   &= z^{n-1}q^{-k} h_k(z,q) \prod_{\substack{i=1 \\ i\neq \ell}}^{n} (1-z^{-1}x_ix_{\ell}^{-1}) (1-q^{\frac{1}{2}}x_\ell yz).
    \end{align}
    Using the fact \eqref{fact of qx}, we have
    \begin{align}\label{add-Lem3.7-1}
        & [x^{\be}]h_k(z,q) \prod_{\substack{i=1\\ i\neq \ell}}^{n} (1-zx_ix_{\ell}^{-1})(1-q^{-\frac{1}{2}}x_\ell y)  \bigg|_{x_\ell \to qx_\ell} \nonumber \\
        &= q^{\be_\ell} [x^\be]h_k(z,q) \prod_{\substack{i=1\\ i\neq \ell}}^{n} (1-zx_ix_{\ell}^{-1})(1-q^{-\frac{1}{2}}x_\ell y)  \nonumber \\
        &= q^{\be_\ell} \bigg[x^\be (1-q^{-\frac{1}{2}}x_\ell^{-1} y)\prod_{\substack{i=1 \\ i\neq \ell}}^{n} (1-zx_i^{-1}x_{\ell}) \bigg] h_{k}(z,q).
    \end{align}
    On the other hand, using \eqref{proof-h-relation} we have
    \begin{align}\label{add-Lem3.7-2}
        & [x^{\be}]h_k(z,q) \prod_{\substack{i=1\\ i\neq \ell}}^{n} (1-zx_ix_{\ell}^{-1})(1-q^{-\frac{1}{2}}x_\ell y)  \bigg|_{x_\ell \to qx_\ell} \nonumber \\
        &= [x^{\be}] z^{n-1}q^{-k} h_k(z,q) \prod_{\substack{i=1 \\ i\neq \ell}}^{n} (1-z^{-1}x_ix_{\ell}^{-1}) (1-q^{\frac{1}{2}}x_\ell yz)  \nonumber \\
       & = z^{n-1}q^{-k} \bigg[x^{\be} (1-q^{\frac12}zyx_\ell^{-1}) \prod_{\substack{i=1 \\ i\neq \ell}}^{n} \Big(1-z^{-1}\frac{x_\ell}{x_i}\Big) \bigg] h_k(z,q).
    \end{align}
    Comparing \eqref{add-Lem3.7-1} and \eqref{add-Lem3.7-2}, we obtain the desired assertion.
\end{proof}

When $m=0$, we only need to compute $[x^\de]h_k(z,q)$.
\begin{lem}
\label{linear relation of gk}
    For any $k\in\bN^{+}$, we have
    \[
    [x^\de] h_{k}(z,q) = (yq^{\frac12})^n \prod_{r=1}^{n}\frac{z^r-q^{k-1}}{1-q^kz^{r-1}} [x^\de] h_{k-1}(z,q).
    \]
\end{lem}
\begin{proof}
 We will use Corollary \ref{F x -ap(r)} in Lemma \ref{qx for gk} with $\be=\de-(e_r+\cdots+e_{n-1})$. For the left side of Lemma \ref{qx for gk}, taking $y=z$ in Corollary \ref{F x -ap(r)}, we can get
    \begin{align*}
        & q^k \bigg[ \frac{x^\de}{x_r\cdots x_{n-1}} (1-q^{-\frac12}yx_n^{-1}) \prod_{i=1}^{n-1} \Big(1-z\frac{x_n}{x_i}\Big) \bigg] h_k(z,q) \\
        &= q^k \bigg[ \frac{x^\de}{x_r\cdots x_{n-1}} \prod_{i=1}^{n-1} \Big(1-z\frac{x_n}{x_i}\Big) \bigg] h_k(z,q) -q^{k-\frac12}y \bigg[ \frac{x^\de}{x_r\cdots x_{n}} \prod_{i=1}^{n-1} \Big(1-z\frac{x_n}{x_i}\Big) \bigg] h_k(z,q) \\
        &= q^k \frac{z^{n-r}-z^n}{1-z} \bigg[ \frac{x^\de}{x_{r+1}\cdots x_n} \bigg] h_k(z,q) - q^{k-\frac12}y \frac{1-z^{n-r+1}}{1-z} \bigg[\frac{x^{\de}}{x_r\cdots x_n} \bigg] h_k(z,q);
    \end{align*}
    Taking $y=z^{-1}$ for the right side, we get
    \begin{align*}
        & z^{n-1} \bigg[x^{\be} (1-q^{\frac12}zyx_n^{-1}) \prod_{i=1}^{n-1} \Big(1-z^{-1}\frac{x_n}{x_i}\Big) \bigg] h_k(z,q) \\
        &= z^{n-1} \bigg[ \frac{x^\de}{x_r\cdots x_{n-1}} \prod_{i=1}^{n-1} \Big(1-z^{-1}\frac{x_n}{x_i}\Big) \bigg] h_k(z,q) - q^{\frac12}z^ny \bigg[ \frac{x^\de}{x_r\cdots x_{n}} \prod_{i=1}^{n-1} \Big(1-z^{-1}\frac{x_n}{x_i}\Big) \bigg] h_k(z,q) \\
       & = z^{n-1} \frac{z^{-(n-r)}-z^{-n}}{1-z^{-1}} \bigg[ \frac{x^\de}{x_{r+1}\cdots x_n} \bigg] h_k(z,q) - q^{\frac12}z^ny \frac{1-z^{-(n-r+1)}}{1-z^{-1}} \bigg[\frac{x^{\de}}{x_r\cdots x_n} \bigg] h_k(z,q) \\
        &= \frac{1-z^r}{1-z} \bigg[ \frac{x^\de}{x_{r+1}\cdots x_n} \bigg] h_k(z,q) - q^{\frac12}y\frac{z^r-z^{n+1}}{1-z} \bigg[\frac{x^{\de}}{x_r\cdots x_n} \bigg] h_k(z,q).
    \end{align*}
    Thus, by Lemma \ref{qx for gk}, we obtain the following relation: for $1 \le r \le n-1$,
    \begin{equation}
    \label{lr [xr]gk and [xr+1]gk}
        \bigg[ \frac{x^\de}{x_{r+1} \cdots x_{n}} \bigg]h_k = \frac{yq^{\frac12}(1-z^{n-r+1})(z^r-q^{k-1})}{(1-q^kz^{n-r})(1-z^r)}\bigg[ \frac{x^\de}{x_r \cdots x_{n}} \bigg]h_k.
    \end{equation}

 Similarly,   for the left side of \eqref{add-h-relation} with $\be=\de$, we have
    \begin{align}\label{add-h-LHS}
        & q^k \bigg[ x^\de (1-q^{-\frac12}yx_n^{-1}) \prod_{i=1}^{n-1} \Big(1-z\frac{x_n}{x_i}\Big) \bigg] h_k(z,q)  \nonumber \\
        &= q^k \bigg[ x^\de \prod_{i=1}^{n-1} \Big(1-z\frac{x_n}{x_i}\Big) \bigg] h_k(z,q) - q^{k-\frac12}y \bigg[ x^\de x_n^{-1} \prod_{i=1}^{n-1} \Big(1-z\frac{x_n}{x_i}\Big) \bigg] h_k(z,q)  \nonumber \\
        &= q^k \frac{1-z^n}{1-z} [x^\de]h_k(z,q) - q^{k-\frac12}y [x^\de x_n^{-1}] h_k(z,q). \quad \text{(by Corollary \ref{cor-k01})}
\end{align}
For the right side of \eqref{add-h-relation} with $\be=\de$, we have
\begin{align}\label{add-h-RHS}
        & z^{n-1} \bigg[x^{\de} (1-q^{\frac12}zyx_n^{-1}) \prod_{i=1}^{n-1} \Big(1-z^{-1}\frac{x_n}{x_i}\Big) \bigg] h_k(z,q)  \nonumber \\
        &= z^{n-1} \bigg[ x^{\de} \prod_{i=1}^{n-1} \Big(1-z^{-1}\frac{x_n}{x_i}\Big) \bigg] h_k(z,q) - q^{\frac12}z^ny \bigg[ x^{\de} x_n^{-1} \prod_{i=1}^{n-1} \Big(1-z^{-1}\frac{x_n}{x_i}\Big) \bigg] h_k(z,q) \nonumber \\
       &= \frac{1-z^n}{1-z} [x^\de] h_k(z,q) - q^{\frac12}z^ny [x^\de x_n^{-1}] h_k(z,q). \quad  \text{(by Corollary \ref{cor-k01})}
    \end{align}
By Lemma \ref{qx for gk}, the two expressions in \eqref{add-h-LHS} and \eqref{add-h-RHS} match each other, and we deduce that
    \begin{align}\label{add-eq-A}
    [x^\de]h_k(z,q) = \frac{yq^\frac12(z^n - q^{k-1})(1-z)}{(1-z^n)(1-q^k)} [x^\de x_n^{-1}] h_k(z,q),
    \end{align}
    which is exactly the relation \eqref{lr [xr]gk and [xr+1]gk} when we set $r=n$ and agree that $\big[\frac{x^\de}{x_{r+1}\cdots x_n} \big]h_k \big|_{r=n} := [x^\de]h_k$.

 Note that $[\frac{x^\de}{x_1\cdots x_n}]h_k = [x^\de]h_{k-1}$. From \eqref{lr [xr]gk and [xr+1]gk} we have
    \begin{align*}
        &[x^\de]h_k = \frac{yq^\frac12(z^n - q^{k-1})(1-z)}{(1-z^n)(1-q^k)} [x^\de x_n^{-1}] h_k \\
        &= \prod_{r=1}^n \frac{yq^{\frac12}(1-z^{n-r+1})(z^r-q^{k-1})}{(1-q^kz^{n-r})(1-z^r)}\bigg[ \frac{x^\de}{x_1 \cdots x_{n}} \bigg]h_k \\
        &= (yq^{\frac12})^n \prod_{r=1}^{n}\frac{z^r-q^{k-1}}{1-q^kz^{r-1}} [x^\de] h_{k-1}. \qedhere
    \end{align*}
\end{proof}

\begin{proof}[Proof of Theorem \ref{thm-CT-SU-N12}]
    From \eqref{General-2-key} and \eqref{add-coeff} we deduce that
    \begin{align}\label{add-N12-proof-1}
&C_0^{(1)}(y,z)= [x_1^0x_2^0\cdots x_{n-1}^0] F_{0,n}(x_1,\dots,x_{n-1},(x_1x_2\cdots x_{n-1})^{-1}) \nonumber \\
&= n! \sum_{k=0}^\infty [x^\de] h_k(z,q).
    \end{align}
By Lemma \ref{linear relation of gk}, we have
    \begin{align}\label{add-hk-h0}
        [x^\de] h_{k}(z,q) = (yq^{\frac12})^n \prod_{r=1}^{n}\frac{z^r-q^{k-1}}{1-q^kz^{r-1}} [x^\de] h_{k-1}(z,q)= (yq^{\frac12})^{nk} \prod_{t=1}^k \prod_{r=1}^{n}\frac{z^r-q^{t-1}}{1-q^tz^{r-1}} [x^\de] h_{0}(z,q).
    \end{align}
We have
    \begin{align}\label{add-h-compute}
  &  [x^\de] h_0(z,q) = [x^0] \prod_{i\ne j}^n \frac{(qx_ix_j^{-1};q)_{\ii}}{(zx_ix_j^{-1};q)_{\ii}} \prod_{i=1}^n \frac{(q^{\frac12}x_iy;q)_{\ii}}{(q^{\frac12}x_iyz;q)_{\ii}} \nonumber \\
  &= [x^0] \prod_{i\ne j}^n \frac{(qx_ix_j^{-1};q)_{\ii}}{(zx_ix_j^{-1};q)_{\ii}} = [x^\de] g_{0,n}(z,q),
    \end{align}
where $g_{m,n}(z,q)$ was defined in \eqref{add-g-defn}.  Note that the $x$-degree of ${(qx_ix_j^{-1};q)_{\ii}}/{(zx_ix_j^{-1};q)_{\ii}}$ is $0$, and ${(q^{\frac12}x_iy;q)_{\ii}}/{(q^{\frac12}x_iyz;q)_{\ii}}$ contains only monomials with nonnegative powers. Thus we have to choose the constant term of ${(q^{\frac12}x_iy;q)_{\ii}}/{(q^{\frac12}x_iyz;q)_{\ii}}$, that is $1$. This explains the last equality in \eqref{add-h-compute}.

Substituting \eqref{g0,n} into \eqref{add-h-compute} and then substituting the result into \eqref{add-N12-proof-1}, we obtain \eqref{eq-thm-SU-N12-C0}.

When $m=1$, note that
\begin{align}\label{add-eq-C}
[x^\de] h_k(z,q) \cdot \sum_{i=1}^{n} x_i = \sum_{i=1}^{n} [x^\de x_i^{-1}] h_k(z,q) = [x^\de x_{n}^{-1}] h_k(z,q)
\end{align}
since $h_k$ is antisymmetric and for $1\leq i<n$, exponents of $x^\de x_i^{-1}$ are not distinct.

Using \eqref{add-coeff} we have
    \begin{align*}
        &C_1^{(1)}(y,z) = [x_1^0x_2^0\cdots x_{n-1}^0] F_{1,n}(x_1,\dots,x_{n-1},(x_1x_2\cdots x_{n-1})^{-1}) \nonumber \\
      & = n! \sum_{k=0}^{\ii} [x^\de] h_k(z,q) \sum_{i=1}^n x_i = n! \sum_{k=0}^{\ii} [x^\de x_{n}^{-1}] h_k(z,q)  \quad \text{(by \eqref{add-eq-C})} \nonumber \\
      &=n! \sum_{k=0}^{\ii} (yq^{\frac{1}{2}})^{nk-1} \frac{(1-z^n)(1-q^k)}{(z^n - q^{k-1})(1-z)} \prod_{t=1}^k \prod_{r=1}^{n}\frac{z^r-q^{t-1}}{1-q^tz^{r-1}}[x^\de] g_{0,n}(z,q).
    \end{align*}
Here for the last equality we used \eqref{add-eq-A}, \eqref{add-hk-h0} and \eqref{add-h-compute}. Substituting  \eqref{g0,n} into it, we obtain \eqref{eq-thm-SU-N12-C1}.
\end{proof}

\begin{proof}[Proof of Theorem \ref{thm-SUN-N-12}]
This follows from \eqref{eq-thm-SU-N12-C0} and \eqref{eq-thm-SU-N12-C1} by setting $z=0$. Note that for the later we also need to replace $k$ by $k+1$.
\end{proof}

Our second goal is to prove Theorem \ref{thm-SUN-N-12-general}.  For convenience, we denote
\begin{align}
&B_{i,k}:= B_{i,k}(y,q) = [x^\de (x_i \cdots x_n)^{-1}] h_k(0,q), \label{def-Bik} \\
&C_{i,m,k}:= C_{i,m,k}(y,q) = [x^\de x_i^{-m}] h_k(0,q).  \label{def-Cimk}
\end{align}
Clearly, we have $[x^\de]h_k = B_{1,k+1}$. From \eqref{lr [xr]gk and [xr+1]gk} with $z=0$, we have
\begin{align}\label{add-Bi-B1}
    B_{i,k} = (-yq^{\frac12})^{i-1}q^{(k-1)(i-1)} B_{1,k}, \quad 1\leq i \leq n.
\end{align}
From the proof of Theorem \ref{thm-CT-SU-N12}, we obtain
\begin{align}
    \label{C-i-0-k}
    B_{1,k+1} = C_{i,0,k} = [x^\de]h_k =\begin{cases}
        \frac{(-y)^{nk}q^{nk^2/2}}{(q;q)_k (q;q)_{\ii}^{n-1}}, & k\geq 0, \\
        0, & k<0.
    \end{cases}
\end{align}
In particular, the case $k\geq 0$ follows from \eqref{add-hk-h0}, \eqref{add-h-compute} and \eqref{g0,n}.

Substituting \eqref{C-i-0-k} into \eqref{add-Bi-B1}, we deduce that
\begin{align}\label{Bik-result}
    B_{i,k} =\begin{cases}
    (-y)^{n(k-1)+i-1} \frac{q^{n(k-1)^2/2+(i-1)(k-\frac12)}}{(q;q)_{k-1}(q;q)_{\ii}^{n-1}}, & k\geq 1, \\
    0, & k\leq 0.
    \end{cases}
\end{align}

From \eqref{General-2-key} and \eqref{add-coeff} we have
\begin{align}
    \label{Cm and Cimk}
    C_m^{(1)}(y,0) = n! \sum_{k = \lceil m/n \rceil}^{\ii} \sum_{i=1}^n C_{i,m,k}(y,q)= n! \sum_{k\in \mathbb{Z}} \sum_{i=1}^n C_{i,m,k}(y,q).
\end{align}
Here the last equality holds since $C_{i,m,k}=0$ when $k<\lceil m/n\rceil$. To compute $C_{i,m,k}$, we need the following fact which is the special case $z=0$ of Lemma \ref{qx for gk}:  For any $\be\in\bZ^n$ and $1\leq \ell \leq n$, we have
    \begin{align}  \label{qxi for hk(0,q)}
        q^{\be_\ell+k} \big( [x^\be]h_k(0,q) - q^{-\frac12}y [x^\be x_\ell^{-1}]h_k(0,q) \big) = (-1)^{n-1} [x^\be x_\ell^n]h_{k-1}(0,q).
    \end{align}

\begin{thm}\label{thm-Cimk}
  Let $1\leq i \leq n$.  We have
    \begin{align}
        C_{i,0,k}(y,q) &= \begin{cases}
            \frac{(-y)^{nk}q^{nk^2/2}}{(q;q)_k (q;q)_{\ii}^{n-1}}, & k \ge 0,\\
            0, & k < 0,
        \end{cases} \label{C_i0k} \\
        C_{i,m,k}(y,q) &= \begin{cases}
            \frac{(-1)^{nk+1} y^{nk-m} q^{(nk^2+m)/2-(n+1-i)k}}{(q;q)_{k-1}(q;q)_{\ii}^{n-1}}, & n-i+1 \le m \le n \text{ and } k > 0, \\
            0, & 0<m<n-i+1 \text{ or } k \le 0.
        \end{cases} \label{C-imk}
    \end{align}
    Furthermore, for other $m$, we have the following recurrence relations:
    \begin{align}\label{rela-C-12}
        C_{i,m,k}(y,q) =\begin{cases}
        q^{\frac12}y^{-1} \big( C_{i,m-1,k}(y,q) + (-1)^n q^{m+i-k-n-1} C_{i,m-n-1,k-1}(y,q) \big), & m>n,  \\
        (-1)^{n-1} q^{k-i-m+1} \big( C_{i,m+n,k+1}(y,q) - yq^{-\frac12} C_{i,m+n+1,k+1}(y,q) \big), & m<0.
        \end{cases}
    \end{align}
\end{thm}

\begin{proof}
    The result of $m=0$ can be obtained directly from \eqref{C-i-0-k}.

    Now, we suppose $1 \le m \le n$. When $1 \le i < n-m+1$, clearly
    \[
    C_{i,m,k} = [x^\de x_i^{-m}]h_k(0,q) = 0
    \]
    by the antisymmetry of $h_k(0,q)$. For $n-m+1 \le i \le n$, we have
    \begin{align}\label{add-Cimk-hk}
       & C_{i,m,k} = [x^\de x_i^{-m}]h_k = [\omega_{n-i+1}^{-1}(x^\de (x_i \cdots x_{n-1})^{-1} x_n^{n-i-m})]h_k(0,q) \notag \\
        &= (-1)^{n-i} [x^\de (x_i \cdots x_{n-1})^{-1} x_n^{n-i-m}]h_k(0,q).
    \end{align}
 In particular, we have
    \begin{align}\label{add-C-B}
    C_{n-m+1,m,k} = (-1)^{m+1} [x^\de (x_{n-m+1} \cdots x_{n})^{-1}]h_k(0,q) = (-1)^{m+1} B_{n-m+1,k}.
    \end{align}
When $n-m+2 \le i \le n$, the power $n-i-m \le -2$. Thus using \eqref{qxi for hk(0,q)} with $\ell=n$ and $\be = \de - (e_i+\cdots+e_{n-1})+(n-i-m+1)e_n$, we deduce that
    \begin{align*}
        &q^{k+n-i-m+1} \big( [x^\de (x_i \cdots x_{n-1})^{-1} x_n^{n-i-m+1}]h_k(0,q) - q^{-\frac12}y [x^\de (x_i \cdots x_{n-1})^{-1} x_n^{n-i-m}]h_k(0,q) \big) \\
    &= (-1)^{n-1} [x^\de (x_i \cdots x_{n-1})^{-1} x_n^{2n-i-m+1}]h_{k-1}(0,q)
    \end{align*}
    Note that $n-i-m \le -2$ and $m\le n$, and thus
    \[
    n-i+1 \le 2n-i-m+1 \le n-1,
    \]
    which implies $[x^\de (x_i \cdots x_{n-1})^{-1} x_n^{2n-i-m+1}]h_{k-1}(0,q)=0$. Hence we get
    \begin{align}\label{add-hk0-rec}
        [x^\de (x_i \cdots x_{n-1})^{-1} x_n^{n-i-m}]h_k(0,q) = y^{-1}q^{\frac12} [x^\de (x_i \cdots x_{n-1})^{-1} x_n^{n-i-m+1}]h_k(0,q).
    \end{align}
 Iterating \eqref{add-hk0-rec} and combining with \eqref{add-Cimk-hk}, we have
    \begin{align}\label{Cimk-Bik-relation}
        &C_{i,m,k} = (-1)^{n-i} [x^\de (x_i \cdots x_{n-1})^{-1} x_n^{n-i-m}]h_k(0,q) \notag \\
        &= (-1)^{n-i} (y^{-1}q^{\frac12})^{m+i-n-1} [x^\de (x_i \cdots x_n)^{-1}]h_k(0,q) = (-1)^{n-i} (y^{-1}q^{\frac12})^{m+i-n-1} B_{i,k},
    \end{align}
    which also holds for $i=n-m+1$ in view of \eqref{add-C-B}. Substituting \eqref{Bik-result} into \eqref{Cimk-Bik-relation}, we obtain the desired formulas for $1 \le m \le n$.

    The cases $m>n$ and $m<0$ in the recurrence relation \eqref{rela-C-12} can be obtained directly from \eqref{qxi for hk(0,q)} by setting $\ell=i,\,\be = \de +(1-m)e_i$ and $\be = \de - (m+n)e_i$, respectively. In particular, for the $m<0$ case we also need to replace $k$ by $k+1$.
\end{proof}

\begin{proof}[Proof of Theorem \ref{thm-SUN-N-12-general}]
First note that
\begin{align}
\langle W_m\rangle^{SU(N)_{-N-1/2}}(x;q)=C_N^{(1)}(x,0).
\end{align}
Using \eqref{Cm and Cimk} and Theorem \ref{thm-Cimk} after replacing $n$ by $N$ and setting
\begin{align}\label{Y-C}
Y_{i,m,k}(x,q)=(q;q)_\infty^{N-1}C_{i,m,k}(x,q),
\end{align}
we obtain the desired assertions.
\end{proof}

\begin{proof}[Proof of Corollary \ref{cor-SUN-N-12-general}]
When $1\leq m \leq n$, using \eqref{C-imk} with $k$ replaced by $k+1$, we obtain
\begin{align}
    \label{C-imk+1-1}
C_{i,m,k+1} (y,q) =\begin{cases} \frac{(-1)^{nk+n+1} y^{nk+n-m} q^{(nk^2-n+m)/2+(k+1)(i-1)}}{(q;q)_k(q;q)_{\ii}^{n-1}}, & \text{$k \ge 0$ and $n-i+1 \le m \le n$}, \\
0, & \text{otherwise.}
\end{cases}
\end{align}
Substituting it into \eqref{Cm and Cimk}, we deduce that
 \begin{align}
        C_m^{(1)}(y,q) &= \frac{(-1)^{n+1}(yq^{\frac12})^{n-m}n!}{(q;q)_{\ii}^{n-1}} \sum_{k \in \bN} \frac{(-y)^{nk} q^{nk^2/2+(n-m)k}\big(1-q^{m(k+1)} \big)}{(q;q)_k(1-q^{k+1})}. \label{Cm>0}
    \end{align}
This proves \eqref{intro-Wn-positive} in view of \eqref{Y-C}.

When $-n+1 \le m \le -1$, then $1 \le m+n < m+n+1 \le n$,
and by \eqref{C-imk+1-1} we obtain for $k\geq 0$ that
\begin{align*}
    &C_{i,m+n,k+1} (y,q) =\begin{cases}
    \frac{(-1)^{nk+n+1} y^{nk-m} q^{(nk^2+m)/2+(k+1)(i-1)}}{(q;q)_k(q;q)_{\ii}^{n-1}}, &  1-i \le m \le 0,\\
    0, & -n+1\leq m \leq -i,
    \end{cases} \\
    &C_{i,m+n+1,k+1} (y,q) =\begin{cases} \frac{(-1)^{nk+n+1} y^{nk-m-1} q^{(nk^2+m)/2+(k+1)(i-1)+\frac12}}{(q;q)_k(q;q)_{\ii}^{n-1}}, &-i \le m \le -1, \\
    0, & -n+1\leq m <-i.
    \end{cases}
\end{align*}
Substituting them into \eqref{rela-C-12}, we obtain
\begin{align}
   & C_{i,m,k} (y,q) = (-1)^{n-1} q^{k-i-m+1} \big( C_{i,m+n,k+1}(y,q) - yq^{-\frac12} C_{i,m+n+1,k+1}(y,q) \big) \notag \\
    &= \begin{cases}
        \frac{(-1)^{nk+1} y^{nk-m} q^{(nk^2-m)/2-mk}}{(q;q)_k(q;q)_{\ii}^{n-1}}, & m=-i,k \ge 0,\\
        0, & \text{otherwise}.
    \end{cases} \label{Cimk-m<0-1}
\end{align}
Then by \eqref{Cm and Cimk} we have
\begin{align}
    &C_m^{(1)}(y,0) = n! \sum_{k \in \bN} \sum_{i=1}^n C_{i,m,k}(y,q) = n! \sum_{k \in \bN} C_{-m,m,k}(y,q)  \nonumber \\
    &= \frac{-(yq^{\frac12})^{-m}n!}{(q;q)_{\ii}^{n-1}} \sum_{k\in\bN} \frac{(-1)^{nk} y^{nk} q^{nk^2/2-mk}}{(q;q)_k}.
\end{align}
This proves \eqref{intro-Wn-negative}  in view of \eqref{Y-C}.

When $m=-n$, by \eqref{rela-C-12} we have
\[
C_{i,-n,k}(y,q) = (-1)^{n-1} q^{k-i+n+1} \big( C_{i,0,k+1}(y,q) - yq^{-\frac12} C_{i,1,k+1}(y,q) \big).
\]
From \eqref{C-imk} we see that $C_{i,1,k+1} \ne 0$ only when $n-i+1 \le 1$, i.e., $i=n$. Thus, from \eqref{Cm and Cimk} we have
\begin{align*}
    &C_{-n}^{(1)}(y,0) = n! \sum_{k \ge -1} \sum_{i=1}^n C_{i,-n,k}(y,q) \nonumber \\
    &= (-1)^{n-1} n! \sum_{k \ge 0} \bigg( \sum_{i=1}^n q^{k-i+n} C_{i,0,k}(y,q) -yq^{k-\frac12} C_{n,1,k} (y,q) \bigg) \\
    &= \frac{(-1)^{n-1}n!}{(q;q)_{\ii}^{n-1}} \bigg( \sum_{k \ge 0} \frac{(-y)^{nk}q^{nk^2/2+k}}{(q;q)_k}\sum_{i=1}^n q^{i-1} - \sum_{k > 0}yq^{k-\frac12} \frac{(-1)^{nk+1} y^{nk-1} q^{(nk^2+1)/2-k}}{(q;q)_{k-1}}  \bigg) \\
    &= \frac{(-1)^{n+1}n!}{(q;q)_{\ii}^{n-1}} \bigg( \sum_{k \ge 0} \frac{(-y)^{nk}q^{nk^2/2+k}}{(q;q)_k}\frac{1-q^n}{1-q} + \sum_{k \ge 0} \frac{(-y)^{nk+n} q^{n(k+1)^2/2}}{(q;q)_{k}}  \bigg).
\end{align*}
Here we replaced $k$ by $k+1$ for the second sum to get the last equality. This proves \eqref{intro-Wn-negative-n}  in view of \eqref{Y-C}.

By \eqref{rela-C-12} we have
\begin{align}
    C_{i,n+1,k}(y,q) &= q^{\frac12}y^{-1} \big( C_{i,n,k}(y,q) +(-1)^n q^{i-k}C_{i,0,k-1}(y,q) \big),\\
    C_{i,-n-1,k}(y,q) &= (-1)^{n-1}q^{k-i+n+2} \big( C_{i,-1,k+1}(y,q) - yq^{-\frac12}C_{i,0,k+1}(y,q) \big).
\end{align}
Thus, by \eqref{Cm and Cimk}, \eqref{C_i0k} and \eqref{C-imk+1-1} we have
\begin{align*}
    &C_{n+1}^{(1)}(y,0) = q^{\frac12}y^{-1} n! \sum_{k\in\bZ} \sum_{i=1}^n \big( C_{i,n,k}(y,q) +(-1)^n q^{i-k}C_{i,0,k-1}(y,q) \big) \\
    &=\frac{q^{\frac12}y^{-1} n!}{(q;q)_{\ii}^{n-1}} \sum_{k\ge0} \sum_{i=1}^n \bigg( \frac{(-1)^{nk+n+1}y^{nk}q^{nk^2/2+(k+1)(i-1)}}{(q;q)_k} +(-1)^nq^{i-k-1} \frac{(-y)^{nk}q^{nk^2/2}}{(q;q)_k} \bigg)\\
    &= \frac{(-1)^nn!}{(q;q)_{\ii}^{n-1}} \sum_{k\ge0} \frac{(-y)^{nk-1} q^{(nk^2+1)/2}}{(q;q)_k} \bigg( \frac{1-q^{n(k+1)}}{1-q^{k+1}} - q^{-k} \frac{1-q^n}{1-q} \bigg).
\end{align*}
This proves \eqref{Cm>0-2}.

Similarly, by \eqref{Cm and Cimk},  \eqref{C_i0k} and  \eqref{Cimk-m<0-1} we have
\begin{align*}
   & C_{-n-1}^{(1)}(y,0) = (-1)^{n-1} n! \sum_{k\in\bZ} \sum_{i=1}^nq^{k-i+n+2} \big( C_{i,-1,k+1}(y,q) - yq^{-\frac12}C_{i,0,k+1}(y,q) \big) \\
    &= (-1)^{n-1} n! \sum_{k\in\bZ} \big( q^{k+n+1} C_{1,-1,k+1}(y,q) -\sum_{i=1}^n yq^{k-i+n+\frac32} C_{i,0,k+1}(y,q) \big) \\
    &= \frac{(-1)^{n-1} n!}{(q;q)_{\ii}^{n-1}} \sum_{k \ge0} \bigg( q^{k+n} \frac{(-y)^{nk+1}q^{(nk^2+1)/2+k}}{(q;q)_k} + \frac{(-y)^{nk+1}q^{(nk^2+1)/2+k}}{(q;q)_k}\sum_{i=1}^nq^{i-1} \bigg)\\
    &= \frac{(-1)^{n-1} n!}{(q;q)_{\ii}^{n-1}} \sum_{k \ge0} \frac{(-y)^{nk+1}q^{(nk^2+1)/2+k}}{(q;q)_k} \bigg( q^{k+n} + \frac{1-q^n}{1-q} \bigg).
\end{align*}
This proves \eqref{Cm<0-2}.
\end{proof}

\subsection{The ${\rm SU}(N)_{-N-1}$ conjecture}
We first establish the following new constant term identity which generalizes Theorem \ref{thm-SUN-N-1}.
\begin{thm}\label{thm-CT-SU-N1}
Let
   \begin{align}\label{eq-thm-CT-SU-N1}
   C_m^{(2)}(y,z):=[(x_1x_2\cdots x_{n-1})^0] F_{m,n}^{(2)}(x_1,x_2,\dots,x_{n-1},(x_1x_2\dots x_{n-1})^{-1})
   \end{align}
where
\begin{align}\label{F2-defn}
    F_{m,n}^{(2)}(x):=\prod_{i\ne j}^{n} \frac{(x_ix_j^{-1};q)_\infty}{(x_ix_j^{-1}z;q)_\infty} \prod_{i=1}^{n} \frac{(q^{\frac12}x_iy,q^{\frac12}x_i^{-1}y^{-1};q)_{\infty}}{(q^{-\frac12}x_iyz,q^{\frac12}x_i^{-1}y^{-1}z;q)_{\infty}}\sum_{i=1}^n x_i^m.
\end{align}
   We have
   \begin{align}
   &C_0^{(2)}(y,z)=n!n  \frac{(z;q)_{\ii}^{n}}{(q;q)_{\ii}^{n-1} (qz^{n};q)_{\ii}(z;z)_{n}} \sum_{k\in \mathbb{Z}} (-1)^{nk} y^{nk} q^{nk^2/2} \prod_{i=1}^{n} \frac{(z^iq^{-k},z^iq^{k+1};q)_{\ii}}{(z^{i+1},z^{i-1}q;q)_{\ii}}, \label{eq-thm-CT-SU-N1-C0} \\
   &C_1^{(2)}(y,z)=n!y^{-1}q^{1/2}\frac{(z;q)_\infty^{n}(1-z^n)}{(q;q)_\infty^{n-1}(1-z)(qz^n;q)_\infty (z;z)_n} \nonumber \\
   &\qquad \qquad  \qquad  \times \sum_{k\in \mathbb{Z}} (-1)^{nk}y^{nk}q^{nk^2/2} \frac{1-zq^k}{z^n-q^k}\prod\limits_{i=1}^n \frac{(z^iq^{-k}, z^iq^{k+1};q)_\infty}{(z^{i+1},z^{i-1}q;q)_\infty}.
   \label{eq-thm-CT-SU-N1-C1}
    \end{align}
\end{thm}
Note that
\begin{align}
\label{tran-2.4}
   C_m^{(2)}(y,z)= \sum_{k \in \bZ} [(x_1x_2\cdots x_n)^k] F_{m,n}^{(2)}(x_1,\dots,x_n).
\end{align}
Clearly, we have
\begin{align*}
    & [x^k] F_{m,n}^{(2)} (x) = [x^0] x^{-k} \prod_{i\ne j}^{n} \frac{(x_ix_j^{-1};q)_\infty}{(x_ix_j^{-1}z;q)_\infty} \prod_{i=1}^{n} \frac{(q^{\frac12}x_iy,q^{\frac12}x_i^{-1}y^{-1};q)_{\infty}}{(q^{-\frac12}x_iyz,q^{\frac12}x_i^{-1}y^{-1}z;q)_{\infty}} \sum_{i=1}^{n} x_i^m \\
    &= n! [x^\de] x^{-k} \prod_{i<j}^{n} (x_i-x_j)\prod_{i\ne j}^{n} \frac{(qx_ix_j^{-1};q)_\infty}{(x_ix_j^{-1}z;q)_\infty} \prod_{i=1}^{n} \frac{(q^{\frac12}x_iy,q^{\frac12}x_i^{-1}y^{-1};q)_{\infty}}{(q^{-\frac12}x_iyz,q^{\frac12}x_i^{-1}y^{-1}z;q)_{\infty}}\sum_{i=1}^n x_i^m \\
    &= n! [x^0] x^{-k} \prod_{i< j}^{n} \frac{(qx_ix_j^{-1},x_jx_i^{-1};q)_\infty}{(x_ix_j^{-1}z,x_jx_i^{-1}z;q)_\infty} \prod_{i=1}^{n} \frac{(q^{\frac12}x_iy,q^{\frac12}x_i^{-1}y^{-1};q)_{\infty}}{(q^{-\frac12}x_iyz,q^{\frac12}x_i^{-1}y^{-1}z;q)_{\infty}}\sum_{i=1}^n x_i^m .
\end{align*}
Set $y=x_{n+1}^{-1}q^{\frac12}$, then we can get
\begin{align*}
    & [x^k] F_{m,n}^{(2)}(x_1,\dots,x_{n}) \\
    &= n! [x^0] x^{-k} \prod_{i<j}^{n+1} \frac{(qx_ix_j^{-1};q)_\infty}{(x_ix_j^{-1}z;q)_\infty}  \frac{(x_jx_i^{-1};q)_{\infty}}{(x_jx_i^{-1}z;q)_{\infty}} \sum_{i=1}^n x_i^m \\
    &= n! \sum_{t\in\bZ} x_{n+1}^t [x^0 x_{n+1}^t] x^{-k} \prod_{i<j}^{n+1} \frac{(qx_ix_j^{-1};q)_\infty}{(x_ix_j^{-1}z;q)_\infty}  \frac{(x_jx_i^{-1};q)_{\infty}}{(x_jx_i^{-1}z;q)_{\infty}} \sum_{i=1}^n x_i^m  \\
   & = n! \sum_{r=1}^n \sum_{t\in\bZ} (y^{-1}q^{\frac12})^t [x^0 x_{n+1}^0] (x \cdot x_{n+1})^{-k} x_{n+1}^{k-t} x_r^m \prod_{i<j}^{n+1} \frac{(qx_ix_j^{-1};q)_\infty}{(x_ix_j^{-1}z;q)_\infty}  \frac{(x_jx_i^{-1};q)_{\infty}}{(x_jx_i^{-1}z;q)_{\infty}}.
\end{align*}
Since there are now $n+1$ variables, for convenience, from here till the end of this subsection, we reset the following notations:
\begin{itemize}
    \item $x := x_1 \cdots x_{n+1}$,
    \item $\de := (n,n-1,\dots,0) \in \bZ^{n+1}$,
    \item $x^{\be} := x_1^{\be_1} \cdots x_{n+1}^{\beta_{n+1}}$ if $\be=(\beta_1,\dots,\beta_{n+1}) \in \bZ^{n+1}$.
\end{itemize}
We have
\begin{align*}
    & [(x_1\cdots x_n)^k] F_{m,n}^{(2)}(x_1,\dots,x_{n}) \\
    &= n! \sum_{r=1}^n \sum_{t\in\bZ} (y^{-1}q^{\frac12})^t [x^0] x^{-k} x_{n+1}^{k-t} x_r^m \prod_{i<j}^{n+1} \frac{(qx_ix_j^{-1};q)_\infty}{(x_ix_j^{-1}z;q)_\infty}  \frac{(x_jx_i^{-1};q)_{\infty}}{(x_jx_i^{-1}z;q)_{\infty}}.
\end{align*}
Note  that $\prod_{i<j}^{n+1} (qx_ix_j^{-1};q)_\infty(x_jx_i^{-1};q)_{\infty}/((x_ix_j^{-1}z;q)_\infty (x_jx_i^{-1}z;q)_{\infty})$ only contains  homogeneous monomials  of degree $0$, and
\[
[x^0] x^{-k} x_{n+1}^{k-t} x_r^m \prod_{i<j}^{n+1} \frac{(qx_ix_j^{-1};q)_\infty}{(x_ix_j^{-1}z;q)_\infty}  \frac{(x_jx_i^{-1};q)_{\infty}}{(x_jx_i^{-1}z;q)_{\infty}}
\]
will vanish unless $-(n+1)k+k-t+m=0$, i.e., $t=m-nk$. Now we have
\begin{align}
    & [(x_1\cdots x_n)^k] F_{m,n}^{(2)}(x_1,\dots,x_{n}) \notag \\
    &= n! \sum_{r=1}^n (y^{-1}q^{\frac12})^{m-nk} [x^0] x^{-k} x_{n+1}^{(n+1)k-m} x_r^m \prod_{i<j}^{n+1} \frac{(qx_ix_j^{-1};q)_\infty}{(x_ix_j^{-1}z;q)_\infty}  \frac{(x_jx_i^{-1};q)_{\infty}}{(x_jx_i^{-1}z;q)_{\infty}} \notag \\
    &= n! (y^{-1}q^{\frac12})^{m-nk}  \sum_{r=1}^n [x^\de] x_{n+1}^{(n+1)k-m} x_r^m g_{k,n+1} (z,q) \nonumber \\
    &= n! (y^{-1}q^{\frac12})^{m-nk}  \sum_{r=1}^n [x^\de x_{n+1}^{m-(n+1)k} x_r^{-m} ] g_{k,n+1} (z,q) \label{Cm-r-n+1}  \\
    &= (-1)^nn! (y^{-1}q^{\frac12})^{m-nk}  \sum_{r=1}^n [x^{\de+1} x_1^{m-(n+1)(k+1)} x_{r+1}^{-m} ] g_{k,n+1} (z,q)  \nonumber \\
    &= (-1)^nn!  (y^{-1}q^{\frac12})^{m-nk}  \sum_{r=2}^{n+1} [x^\de] x_1^{(n+1)(k+1)-m} x_r^m g_{k+1,n+1}(z,q). \label{add-eq-D}
\end{align}
Here for the penultimate line, we used the fact that $g_{k,n+1}$ is antisymmetric and
\begin{align}
\omega_{n+1} (x^{\de+1} x_1^{m-(n+1)(k+1)} x_{r+1}^{-m}) = x^\de x_{n+1}^{m-(n+1)k} x_r^{-m}.
\end{align}

\begin{proof}[Proof of Theorem \ref{thm-CT-SU-N1}]
(1) When $m=0$ we have
\begin{align}\label{xk-0-coeff-result}
   &[(x_1x_2\cdots x_n)^k] F_{0,n}^{(2)}(x_1,\dots,x_n) \nonumber \\
   &= n!n (-1)^n  y^{nk} q^{-nk/2} [x^\de] x_1^{(n+1)(k+1)} g_{k+1,n+1}(z,q) \quad \text{(by \eqref{add-eq-D})}  \nonumber \\
   &=n!n (-1)^n y^{nk} q^{-nk/2} A_{1,k+1,n+1}(z,q).  \quad \text{(by \eqref{A-defn})}
\end{align}
Substituting Theorem \ref{A1} into it, and then substituting the result into \eqref{tran-2.4}, we obtain \eqref{eq-thm-CT-SU-N1-C0}.

(2) When $m=1$ we have
    \begin{align}\label{xk-coeff}
       & [(x_1\cdots x_n)^k] F_{1,n}^{(2)}(x_1,\dots,x_n) \nonumber \\
       &= n! (-1)^n  y^{nk-1} q^{\frac{1-nk}{2}} \sum_{r=2}^{n+1} [x^\de] x_1^{(n+1)(k+1)-1} x_r g_{k+1,n+1}(z,q)  \quad \text{(by \eqref{add-eq-D})}\nonumber \\
       & = n! (-1)^n  y^{nk-1} q^{\frac{1-nk}{2}} \sum_{r=2}^{n+1} [x^\de x_1^{1-(n+1)(k+1)} x_r^{-1}]  g_{k+1,n+1}(z,q)  \nonumber  \\
       & = n! (-1)^n  y^{nk-1} q^{\frac{1-nk}{2}}  [x^\de x_1^{1-(n+1)(k+1)}x_{n+1}^{-1}]  g_{k+1,n+1}(z,q).
    \end{align}
Here for the last equality we used the antisymmetry of $g_{k+1,n+1}(z,q)$ to deduce that
\begin{align}
[x^\de x_1^{1-(n+1)(k+1)}x_r^{-1}] g_{k+1,n+1}(z,q)=0, \quad 2\leq r \leq n.
\end{align}

We now consider Lemma \ref{qx for gm,n} with $n+1$ variables and $\be = \de +(1-(n+1)(k+1))e_1-e_{n+1}$. For the left side of \eqref{eq-lem-gmn}, from \eqref{F x de x1 k and xn -1} we have
    \begin{align}\label{N1-thm2-proof-final-1}
        & q^{k} \bigg[ x^\de x_1^{1-(n+1)(k+1)} x_{n+1}^{-1}  \prod_{i=1}^n \Big(1-z\frac{x_{n+1}}{x_i}\Big)  \bigg] g_{k+1,n+1}(z,q) \\
       & = q^{k} [ x^\de x_1^{1-(n+1)(k+1)} x_{n+1}^{-1}] g_{k+1,n+1}(z,q) - q^{k} \frac{z-z^{n+1}}{1-z}[ x^\de x_1^{-(n+1)(k+1)}] g_{k+1,n+1} (z,q). \nonumber
    \end{align}
For the right side of \eqref{eq-lem-gmn}, from \eqref{F x de x1 k and xn -1} we have
    \begin{align}\label{N1-thm2-proof-final-2}
        & z^n  \bigg[ x^\de x_1^{1-(n+1)(k+1)} x_{n+1}^{-1}  \prod_{i=1}^n \Big(1-z^{-1}\frac{x_{n+1}}{x_i}\Big) \bigg] g_{k+1,n+1}(z,q)   \\
       & = z^n [x^\de x_1^{1-(n+1)(k+1)} x_{n+1}^{-1}]g_{k+1,n+1}(z,q) - \frac{1-z^n}{1-z} [x^\de x_1^{-(n+1)(k+1)}] g_{k+1,n+1}(z,q). \nonumber
    \end{align}
Combining \eqref{N1-thm2-proof-final-1} and \eqref{N1-thm2-proof-final-2}, from Lemma \ref{qx for gm,n} we deduce that
\begin{align}\label{N1-thm2-proof-relation}
   & [x^\de x_1^{1-(n+1)(k+1)} x_{n+1}^{-1}] g_{k+1,n+1}(z,q) = \frac{(1-z^n)(1-zq^k)}{(1-z)(z^n-q^k)} [x^\de x_1^{-(n+1)(k+1)}] g_{k+1,n+1}(z,q) \nonumber \\
   &= \frac{(1-z^n)(1-zq^k)}{(1-z)(z^n-q^k)} A_{1,k+1,n+1}(z,q).
\end{align}
 Substituting \eqref{N1-thm2-proof-relation} into \eqref{xk-coeff}, we obtain
 \begin{align}\label{xk-coeff-result}
       & [(x_1\cdots x_n)^k] F_{1,n}^{(2)}(x_1,\dots,x_n) \nonumber \\
       &=  n! (-1)^n  y^{nk-1} q^{\frac{1-nk}{2}}\frac{(1-z^n)(1-zq^k)}{(1-z)(z^n-q^k)} A_{1,k+1,n+1}(z,q).
 \end{align}
Substituting Theorem \ref{A1} into it \eqref{xk-coeff-result}, and then substituting the result into  \eqref{tran-2.4},  we obtain \eqref{eq-thm-CT-SU-N1-C1}.
\end{proof}

\begin{proof}[Proof of Theorem \ref{thm-SUN-N-1}]
This follows from Theorem \ref{thm-CT-SU-N1} by setting $z=0$. For the later formula we also need to replace $k$ by $k+1$.
\end{proof}

We now turn to prove Theorem \ref{thm-SUN-N1-general}, which is equivalent to
\begin{thm}\label{thm-Cm2}
    Let $t=\lfloor m/(n+1) \rfloor$. We have
    \begin{align}
        C_m^{(2)}(y,0) = \begin{cases}
            \frac{n!}{(q;q)_{\ii}^n} y^{-m} q^{(m-n)t+m/2} \frac{1-q^{nt}}{1-q^t} \sum_{k\in\bZ} (-1)^{nk} y^{nk} q^{nk^2/2 -mk}, & m \equiv 0 \pmod{n+1}, \\
            -\frac{n!}{(q;q)_{\ii}^n} y^{-m} q^{mt+m/2} \sum_{k\in\bZ} (-1)^{nk} y^{nk} q^{nk^2/2 -mk}, & m \not\equiv 0 \pmod{n+1}.
        \end{cases}
    \end{align}
\end{thm}
\begin{proof}
    For any $\be \in \bZ^{n+1}$, by Lemma \ref{qx for gm,n} with $z=0$ we have
    \begin{align}
        [x^\be] g_{m,n+1}(0,q) = (-1)^n q^{-\be_{n+1}-m} [x^\be x_{n+1}^{n+1}] g_{m-1,n+1}(0,q). \label{qx-gm,n+1-1}
    \end{align}
Replacing $m$ by $m+1$ and $\beta$ by $\beta-(n+1)e_{n+1}$, this can be written as the equivalent form:
    \begin{align}
        [x^\be] g_{m,n+1}(0,q) = (-1)^n q^{\be_{n+1}-n+m} [x^{\be}x_{n+1}^{-n-1}] g_{m+1,n+1}(0,q). \label{qx-gm,n+1-2}
    \end{align}

Suppose $m=t(n+1)+r$ with $r \in [0,n]$, we claim that
\begin{align}
       & [x^\de x_i^{-m} x_{n+1}^{m-(n+1)k}] g_{k,n+1}(0,q) \nonumber \\
       &= (-1)^{(t-k)n} q^{(nk(k+1)+(n+2)t^2-nt)/2+rt-mk} [x^\de x_i^{-m} x_{n+1}^r] g_{t,n+1}(0,q). \label{gk,n+1-gt,n+1}
\end{align}
In fact, we prove this claim by discussing the cases $t<k$, $t=k$ and $t>k$ separately.
 When $t=k$, \eqref{gk,n+1-gt,n+1} is obvious.
 When $t<k$, by \eqref{qx-gm,n+1-1}, we have
    \begin{align*}
        & [x^\de x_i^{-m} x_{n+1}^{m-(n+1)k}] g_{k,n+1}(0,q) = [x^\de x_i^{-m} x_{n+1}^{(t-k)(n+1)+r}] g_{k,n+1}(0,q) \\
        &= (-1)^n q^{-(t-k)(n+1)-r-k} [x^\de x_i^{-m} x_{n+1}^{(t-k+1)(n+1)+r}] g_{k-1,n+1}(0,q) \\
        &= (-1)^n q^{-(t-k)(n+1)-r-k} \times (-1)^n q^{-(t-k+1)(n+1)-r-k+1} [x^\de x_i^{-m} x_{n+1}^{(t-k+2)(n+1)+r}] g_{k-2,n+1}(0,q) \\
        &= (-1)^{2n} q^{-(n+1)[(t-k)+(t-k+1)]-2r-2k+(0+1)} [x^\de x_i^{-m} x_{n+1}^{(t-k+2)(n+1)+r}] g_{k-2,n+1}(0,q) \\
        &= (-1)^{(k-t)n} q^{-(n+1)[(t-k)+\cdots+(-1)]-(k-t)(r+k)+[0+\cdots+(k-t-1)]} [x^\de x_i^{-m} x_{n+1}^r] g_{t,n+1}(0,q) \\
        &= (-1)^{(t-k)n} q^{(n+1)(k-t)(k-t+1)/2-(k-t)(r+k)+(k-t)(k-t-1)/2} [x^\de x_i^{-m} x_{n+1}^r] g_{t,n+1}(0,q) \\
       &= (-1)^{(t-k)n} q^{(nk(k+1)+(n+2)t^2-nt)/2+rt-mk} [x^\de x_i^{-m} x_{n+1}^r] g_{t,n+1}(0,q).
    \end{align*}
    Similarly, for $t>k$, by \eqref{qx-gm,n+1-2}, we have
    \begin{align*}
        & [x^\de x_i^{-m} x_{n+1}^{m-(n+1)k}] g_{k,n+1}(0,q) \\
        &= (-1)^n q^{(t-k)(n+1)+r-n+k} [x^\de x_i^{-m} x_{n+1}^{(t-k-1)(n+1)+r}] g_{k+1,n+1}(0,q) \\
        &= (-1)^{2n} q^{(n+1)[(t-k)+(t-k-1)]+2(r-n+k)+(0+1)} [x^\de x_i^{-m} x_{n+1}^{(t-k-2)(n+1)+r}] g_{k+2,n+1}(0,q) \\
        &= (-1)^{(t-k)n} q^{(n+1)[(t-k)+\cdots+1]+(t-k)(r-n+k)+[0+\cdots+(t-k-1)]} [x^\de x_i^{-m} x_{n+1}^r] g_{t,n+1}(0,q) \\
        &= (-1)^{(t-k)n} q^{[nk(k+1)+(n+2)t^2-nt]/2+rt-mk} [x^\de x_i^{-m} x_{n+1}^r] g_{t,n+1}(0,q).
    \end{align*}

Since $g_{t,n+1}$ is antisymmetric, the right side of \eqref{gk,n+1-gt,n+1} vanishes unless $r=0$ or $r=n+1-i$.

    Now we compute $[x^\de x_i^{-m} x_{n+1}^r] g_{t,n+1}(0,q)$ when $r=0$ or $n+1-i$. For $r=0$,
    \begin{align}
        [x^\de x_i^{-m} x_{n+1}^r] g_{t,n+1}(0,q) = [x^\de x_i^{-t(n+1)}] g_{t,n+1}(0,q) = A_{i,t,n+1}(0,q). \label{Cm-r=0}
    \end{align}
    For $r=n+1-i$, we have
    \begin{align}
        [x^\de x_i^{-(t+1)(n+1)+i}x_{n+1}^{n+1-i}] g_{t,n+1}(0,q) = -[x^\de x_{n+1}^{-t(n+1)}] g_{t,n+1}(0,q) = -A_{n+1,t,n+1}(0,q). \label{Cm-r=n+1-i}
    \end{align}
In fact, if we interchange $x_i$ and $x_{n+1}$, then $x^\delta x_i^{-(t+1)(n+1)+i}x_{n+1}^{n+1-i}$ becomes $x^{\delta}x_{n+1}^{-t(n+1)}$ and $g_{t,n+1}(0,q)$ becomes $-g_{t,n+1}(0,q)$.

From \eqref{add-Cor2.9-relation} and \eqref{eq-thmA1}, we have
    \begin{align}
        A_{l,t,n+1}(0,q) = q^{(l-1)t} A_{1,t,n+1}(0,q) = (-1)^{nt} q^{(l-1)t+nt(t-1)/2} (q;q)_{\ii}^{-n} \label{Cm-Al,n+1}
    \end{align}
    with $1 \le l \le n+1$. Thus, we have
    \begin{align*}
        C_m^{(2)}(y,0) &=  \sum_{k\in\bZ} [(x_1 \cdots x_n)^k] F_{m,n}^{(2)} (x_1,\dots,x_n) \quad \text{(by \eqref{tran-2.4})} \\
        &= n! (y^{-1}q^{\frac12})^m \sum_{k\in\bZ} (yq^{-\frac12})^{nk} \sum_{i=1}^n [x^\de x_{n+1}^{m-(n+1)k} x_i^{-m}] g_{k,n+1}(0,q) \quad \text{(by \eqref{Cm-r-n+1})} \\
        &= n! (y^{-1}q^{\frac12})^m \sum_{k\in\bZ} (yq^{-\frac12})^{nk} (-1)^{(t-k)n} q^{[nk(k+1)+(n+2)t^2-nt]/2+rt-mk} \\
        & \qquad \times \sum_{i=1}^n [x^\de x_i^{-m} x_{n+1}^r] g_{t,n+1}(0,q)  \quad\text{(by \eqref{gk,n+1-gt,n+1})}.
    \end{align*}

If $m \equiv 0 \pmod{n+1}$, then $r=0$. By \eqref{Cm-r=0} and \eqref{Cm-Al,n+1}, we have
    \begin{align*}
        C_m^{(2)}(y,0) &=  n! (y^{-1}q^{\frac12})^m \sum_{k\in\bZ} (yq^{-\frac12})^{nk} (-1)^{(t-k)n} q^{[nk(k+1)+(n+2)t^2-nt]/2+rt-mk} \\
        & \qquad \times \sum_{i=1}^n (-1)^{nt} q^{(i-1)t+nt(t-1)/2} (q;q)_{\ii}^{-n} \\
        &= \frac{n!}{(q;q)_{\ii}^n} y^{-m} q^{(m-n)t+m/2} \frac{1-q^{nt}}{1-q^t} \sum_{k\in\bZ} (-1)^{nk} y^{nk} q^{nk^2/2-mk}.
    \end{align*}
If $m \not\equiv 0 \pmod{n+1}$, then $1\leq r\leq n$. By \eqref{Cm-r=n+1-i}, \eqref{Cm-Al,n+1} and the aforementioned fact that $[x^\de x_i^{-m}x_{n+1}^r] g_{t,n+1}$ is zero unless $i=n+1-r$, we deduce that
    \begin{align*}
        C_m^{(2)}(y,0) &=  -n! (y^{-1}q^{\frac12})^m \sum_{k\in\bZ} (yq^{-\frac12})^{nk} (-1)^{(t-k)n} q^{[nk(k+1)+(n+2)t^2-nt]/2+rt-mk} \\
        & \qquad \times (-1)^{nt} q^{nt+nt(t-1)/2} (q;q)_{\ii}^{-n} \\
        &= -\frac{n!}{(q;q)_{\ii}^n} y^{-m} q^{m/2+mt} \sum_{k\in\bZ} (-1)^{nk} y^{nk} q^{nk^2/2-mk}. \qedhere
    \end{align*}
\end{proof}

\begin{proof}[Proof of Theorem \ref{thm-SUN-N1-general}]
This follows from Theorem \ref{thm-Cm2} upon replacing $n$ and $N$.
\end{proof}

\begin{proof}[Proof of Corollary \ref{cor-SU34-conj}]
This follows from Theorem \ref{thm-SUN-N1-general} with $N=3$. Note that we need to replace $k$ by $k+m$ in \eqref{eq-thm-SUN-N1-general} and use \eqref{jtpi1} as well.
\end{proof}

Alert readers might have noticed that the first assertion of Conjecture \ref{conj-SU34} depends on $m$ modulo 3. This is not easily seen (though not difficult to verify) from  \eqref{SU33-Wk}. This motivates us to give the following recurrence relation for $\langle W_m\rangle^{{\rm SU}(3)_{-4}}$ from which Conjecture \ref{conj-SU34} can be observed.
\begin{cor}\label{cor-SU34}
We have
\begin{align}
\langle W_{m+3}\rangle^{{\rm SU}(3)_{-4}}(x;q) &= \begin{cases}
\frac{q^{\frac{m}{4}}}{(1+q^{\frac{m}{4}}+q^{-\frac{m}{4}})}\langle W_{m}\rangle^{{\rm SU}(3)_{-4}}(x;q), & m=4\ap, \\
(1+q^{\frac{m+3}{4}}+q^{\frac{2m+6}{4}}) \langle W_m\rangle^{{\rm SU}(3)_{-4}}(x;q), & m=4\ap+1, \\
-q^{\frac{3m+6}{4}}\langle W_m\rangle^{{\rm SU}(3)_{-4}}(x;q), & m=4\ap+2, \\
-q^{\frac{3m+3}{4}}\langle W_m\rangle^{{\rm SU}(3)_{-4}}(x;q), & m=4\ap+3.
\end{cases}
\end{align}
\end{cor}
\begin{proof}
If $m=4\ap$, then by \eqref{SU34-W} and \eqref{jtpi1} we have
\begin{align*}
\langle &W_{m+3}\rangle^{{\rm SU}(3)_{-4}}(x;q) = \frac{q^{\frac34(m+3)^2-\frac{1}{4}(m+3)}}{(q;q)_\infty}\sum_{k\in\bZ}(-1)^kx^{3k+2(m+3)}q^{\frac32k^2+2(m+3)k}\\
&= \frac{q^{\frac34m^2+\frac{1}{4}m}}{(q;q)_{\ii}}\sum_{k\in\bZ}(-1)^kx^{3(k+2)+2m}q^{\frac32(k+2)^2+2m(k+2)}\\
&= \frac{q^{\frac{m}{4}}}{1+q^{\frac{m}{4}}+q^{-\frac{m}{4}}}\langle W_m\rangle^{{\rm SU}(3)_{-4}}(x;q).
\end{align*}
Using similar arguments, we can obtain the remaining cases.
\end{proof}

\section{Some byproducts and remarks}\label{sec-rem}
From the results stated in Section \ref{sec-intro} and their proofs we can get some constant term identities as byproducts. To be consistent with the theorems there, we still state them in the form of integrals. As indicated in \cite[Section 2.4]{OS24}, the integrals in this section are connected with half indices arising from the ${\rm U}(N)$ CS theories.
\begin{cor}
\label{U(N) -N}
    We have
    \begin{align}
        \frac{(q;q)_{\ii}^N}{N!}\oint \bigg( \prod_{i=1}^N \frac{dt_i}{2\pi it_i} \bigg) \bigg( \prod_{i \ne j}^N (t_it_j^{-1};q)_{\ii} \bigg)\sum_{i=1}^N t_i^m  = \begin{cases} N(q;q)_{\ii}, & m=0, \\
         0, & m\neq 0.
         \end{cases}
    \end{align}
\end{cor}
\begin{proof}
We denote the integral in the left side by $L_m(q)$. Let $t=t_1t_2\cdots t_N$. We write $t_i = t^{1/N}s_i$ with $s_1s_2\cdots s_N=1$. We have
\begin{equation}
\label{ti and si}
    \prod_{i=1}^N \oint \frac{dt_i}{2\pi it_i} = \oint \frac{dt}{2\pi it} \prod_{i=1}^{N-1} \oint \frac{ds_i}{2\pi is_i}.
\end{equation}
    By \eqref{ti and si} we have
    \begin{align}
       L_m(q)= \oint \frac{dt}{2\pi it} t^{m/N}  \oint \bigg( \prod_{i=1}^{N-1} \frac{ds_i}{2\pi is_i} \bigg) \bigg( \prod_{i \ne j}^N (s_is_j^{-1};q)_{\ii} \bigg)\sum_{i=1}^Ns_i^m.
    \end{align}
It is then clear that $L_m(q)=0$ for $m\neq 0$, and using Theorem \ref{thm-SUN-N} we get the desired expression for $L_0(q)$.
\end{proof}

\begin{cor}
\label{U(N) -N-1/2}
    We have
    \begin{align}
        \frac{(q;q)_{\ii}^N}{N!} \oint \bigg( \prod_{i=1}^N \frac{dt_i}{2\pi it_i} \bigg) \prod_{i \ne j}^N (t_it_j^{-1};q)_{\ii} \prod_{i=1}^N (q^{\frac12}t_iy;q)_{\ii} \sum_{i=1}^N t_i^m = (q;q)_{\ii} \sum_{i=1}^N Y_{i,m,0}(y,q)
    \end{align}
where $Y_{i,m,k}(y,q)$ is defined in Theorem \ref{thm-SUN-N-12-general}.
\end{cor}
\begin{proof}
We denote the integral on the left side as $L_m(q)$.  Recall the function $F_{m,N}^{(1)}(x)$ defined in \eqref{def-Fmn-1}. We have
    \begin{align*}
       L_m(q) & = [(t_1 \cdots t_N)^0] F_{m,N}^{(1)} (t_1,\dots, t_N) \big|_{z=0}  \\
        &= N! [t^{\de(N)}] h_0(0,q) \sum_{i=1}^{N} t_i^m = N! \sum_{i=1}^N C_{i,m,0}(y,q).
\end{align*}
Here $h_k(z,q)$ is defined in \eqref{hk-defn} with $n=N$ and $x_i$ replaced by $t_i$, and for the second and last equalities we used \eqref{add-coeff} and \eqref{def-Cimk}, respectively. Then by \eqref{Y-C}, we obtain the desired formula.
\end{proof}

As a consequence, we have
\begin{align}\label{cor4.2-special}
        \frac{(q;q)_{\ii}^N}{N!} \oint \bigg( \prod_{i=1}^N \frac{dt_i}{2\pi it_i} \bigg) \prod_{i \ne j}^N (t_it_j^{-1};q)_{\ii} \prod_{i=1}^N (q^{\frac12}t_iy;q)_{\ii} \sum_{i=1}^N t_i^m =\begin{cases}
            N(q;q)_\infty, &m=0, \\
            0, & m\geq 1.
        \end{cases}
\end{align}
The case $m\geq 1$ is almost trivial and the case $m=0$ can also be proved directly from Corollary \ref{U(N) -N}.  Note that any term in the expansion of the first product $\prod_{i \ne j}^N (t_it_j^{-1};q)_{\ii}$ has degree 0, and any term in the expansion of the second product $\prod_{i=1}^N (q^{\frac12}t_iy;q)_{\ii}\sum_{i=1}^N t_i^m$ has degree $\geq m$. Since $m\geq 0$, the constant term can only occur when we choose the constant terms of both products. It follows that $L_m(q)=0$ when $m>0$ and
 \begin{align}\label{Lq-reduce}
        L_0(q) =N[(t_1t_2\cdots t_N)^0] \prod_{i \ne j}^N (t_it_j^{-1};q)_{\ii}=N\oint \bigg( \prod_{i=1}^N \frac{dt_i}{2\pi it_i} \bigg)\prod_{i \ne j}^N (t_it_j^{-1};q)_{\ii} .
 \end{align}
 Now using Corollary \ref{U(N) -N} we obtain \eqref{cor4.2-special}.

\begin{cor}
\label{U(N) -N-1}
Let $t = \lfloor m/(N+1) \rfloor$.    We have
    \begin{align}
        &\frac{(q;q)_{\ii}^N}{N!} \oint \bigg( \prod_{i=1}^N \frac{dt_i}{2\pi it_i} \bigg) \prod_{i \ne j}^N (t_it_j^{-1};q)_{\ii} \prod_{i=1}^N (q^{\frac12}t_i^{\pm}y^{\pm};q)_{\ii}\sum_{i=1}^N t_i^m \notag \\
        & = \begin{cases}  y^{-m}q^{(m-N)t+m/2} \frac{1-q^{Nt}}{1-q^t}, & m \equiv 0 \pmod{N+1}, \\
         -y^{-m}q^{mt+m/2}, & m \not\equiv 0 \pmod{N+1}.
         \end{cases}
    \end{align}
\end{cor}
\begin{proof}
We denote the integral on the left side as $L_m(q)$. Recall the function $F_{m,N}^{(2)}(x)$ defined in \eqref{F2-defn}. By \eqref{Cm-r-n+1} we have
    \begin{align}\label{UN1-proof}
        L_m(q)= [t^0]F_{m,N}^{(2)}(t_1,t_2,\dots,t_N) \big|_{z=0} = N! y^{-m}q^{m/2} \sum_{i=1}^N [t^{\de(N)}t_{N+1}^m t_i^{-m}] g_{0,N+1}(0,q).
    \end{align}
As in the proof of Theorem \ref{thm-Cm2}, we write $m = t(N+1) +r$ with $0 \le r \le N$. Then by \eqref{gk,n+1-gt,n+1}, we have
    \begin{align}
        [t^{\de(N)}t_{N+1}^m t_i^{-m}] g_{0,N+1}(0,q) = (-1)^{tN} q^{[(N+2)t^2-Nt]/2+rt} [t^{\de(N)}t_i^{-m}t_{N+1}^r] g_{t,N+1}(0,q).
    \end{align}
    When $m \equiv 0 \pmod{N+1}$, i.e., $r=0$, by \eqref{Cm-r=0} and \eqref{Cm-Al,n+1}, we have
    \begin{align*}
        L_m(q) &= N! y^{-m}q^{m/2} \sum_{i=1}^N (-1)^{tN} q^{[(N+2)t^2-Nt]/2} (-1)^{Nt} q^{(i-1)t+Nt(t-1)/2} (q;q)_{\ii}^{-N} \\
        &= \frac{N!}{(q;q)_{\ii}^N} y^{-m} q^{(m-N)t+m/2} \frac{1-q^{Nt}}{1-q^t}.
    \end{align*}
    When $m \not\equiv 0 \pmod{N+1}$, by \eqref{Cm-r=n+1-i}, \eqref{Cm-Al,n+1} and the fact that
    $$[t^{\de(N)}t_i^{-m}t_{N+1}^r]g_{t,N+1}(0,q) = 0 \quad \text{if} \quad i \ne N+1-r,$$
we get
    \begin{align*}
        L_m(q) &= N! y^{-m}q^{m/2} (-1)^{tN} q^{[(N+2)t^2-Nt]/2+rt} [t^{\de(N)}t_{N+1-r}^{-m}t_{N+1}^r] g_{t,N+1}(0,q) \\
        &= N! y^{-m}q^{m/2} (-1)^{tN+1} q^{[(N+2)t^2-Nt]/2+rt} A_{N+1,t,N+1}(0,q) \\
        &= N! y^{-m}q^{m/2} (-1)^{tN+1} q^{[(N+2)t^2-Nt]/2+rt} (-1)^{Nt} q^{Nt+Nt(t-1)/2} (q;q)_{\ii}^{-N} \\
        &= -\frac{N!}{(q;q)_{\ii}^N} y^{-m} q^{mt+m/2}. \qedhere
    \end{align*}
\end{proof}

We close this paper with several remarks. So far we have completed the task of evaluating half-indices and Wilson lines for ${\rm SU}(N)_{-N-k}$ CS theories for the initial cases $k=0,1/2,1$ and confirmed some intriguing conjectures in \cite{OS-24JHEP}.  Okazaki and Smith \cite[Sections 5.4 and 5.5]{OS-24JHEP} also proposed some conjectural formulas for ${\rm SU}(N)_{-N-k}$ with $k=M$ and $M-1/2$. However, the corresponding formulas \cite[Eqs.~(5.35), (5.38) and (5.41)]{OS-24JHEP} are quite complicated and we are not sure whether they can be treated using the method here. Furthermore, in a subsequent work, they \cite{OS24} also considered line defect half-indices of 3d $\mathcal{N}=2$ supersymmetric CS theories with (special) unitary, symplectic, orthogonal and exceptional gauge groups and the companion theory with an adjoint chiral. A number of beautiful formulas expressing half-indices as eta products have been proposed there. We plan to discuss these formulas in forthcoming papers.


\subsection*{Acknowledgements}
This work was supported by the National Key R\&D Program of China (Grant No.\ 2024YFA1014500).

\end{document}